%% file: main.tex
\documentclass[sigconf,screen,nonacm]{acmart}
 
\usepackage[utf8]{inputenc}

\usepackage{booktabs, tabularx, makecell}

\setcopyright{acmcopyright}
\copyrightyear{2022}
\acmYear{2022}
\acmConference[LICS '22]{LICS '22: ACM/IEEE Symposium on Logic in Computer Science}{August 02--05, 2022}{Haifa, Israel}
\let\C\relax
\let\U\relax

\usepackage{newfile}
\usepackage{complexity}
\usepackage{todonotes}
\usepackage{bm}
\usepackage{mathtools}
\usepackage{tikz,tikz-qtree}
\usetikzlibrary{fadings}
\usepackage{thm-restate}

\usetikzlibrary{shapes,backgrounds,calc}

\makeatletter
\tikzset{circle split part fill/.style  args={#1,#2}{%
 alias=tmp@name, %
  postaction={%
    insert path={
     \pgfextra{%
     \pgfpointdiff{\pgfpointanchor{\pgf@node@name}{center}}%
                  {\pgfpointanchor{\pgf@node@name}{east}}%
     \pgfmathsetmacro\insiderad{\pgf@x}
      \fill[#1] (\pgf@node@name.base) ([xshift=-\pgflinewidth]\pgf@node@name.east) arc
                          (0:180:\insiderad-\pgflinewidth)--cycle;
      \fill[#2] (\pgf@node@name.base) ([xshift=\pgflinewidth]\pgf@node@name.west)  arc
                           (180:360:\insiderad-\pgflinewidth)--cycle;            %
         }}}}}  
 \makeatother
 
\usepackage{tikz}
\tikzset{
sibling distance=1cm, level distance=1cm
}
\usepackage[export]{adjustbox}

\input{macros}

\usepackage[capitalise,nameinlink,noabbrev]{cleveref}

\begin{document}
\sloppy %

\title{Ramsey Quantifiers over Automatic Structures: Complexity and Applications
    to Verification}

\author{Pascal Bergstr\"{a}{\ss}er}
\orcid{0000-0002-4681-2149}
\affiliation{%
  \institution{Department of Computer Science, TU Kaiserslautern}
  \city{Kaiserslautern}
  \country{Germany}
}

\author{Moses Ganardi}
\orcid{0000-0002-0775-7781}
\affiliation{%
  \institution{Max Planck Institute for Software Systems (MPI-SWS)}
  \city{Kaiserslautern}
  \country{Germany}
}

\author{Anthony W. Lin}
\orcid{0000-0003-4715-5096}
\affiliation{%
  \institution{Department of Computer Science, TU Kaiserslautern}
  \city{}
  \country{}
}
\affiliation{%
  \institution{Max Planck Institute for Software Systems (MPI-SWS)}
  \city{Kaiserslautern}
  \country{Germany}
}

\author{Georg Zetzsche}
\orcid{0000-0002-6421-4388}
\affiliation{%
  \institution{Max Planck Institute for Software Systems (MPI-SWS)}
  \city{Kaiserslautern}
  \country{Germany}
}

\begin{abstract}
    \input{abstract}

\end{abstract}



\maketitle

\section{Introduction}\label{sec:intro}

\input{intro_sec}

\section{Detailed summary of main results}
\label{sec:results}
\input{results}

\section{Preliminaries}\label{sec:preliminaries}

\input{preliminaries}

\section{Word-automatic structures}\label{sec:words}

\input{word}

\section{Tree-automatic structures}\label{sec:trees}
\subsection{$\EXP$-hardness}

\input{hardness}

\subsection{Tree combs}

\input{combs}

\subsection{Arbitrary relations}\label{sec:non-transitive}

\input{tree}

\subsection{Transitive relations}

\input{transitive}

\subsection{Co-transitive relations}

\input{cotransitive}

\section{Applications}\label{sec:generalized}

\input{generalized-buchi}

\section{Unranked tree-automatic structures}\label{sec:unranked}
\input{unranked}

\input{conc}

\begin{acks}
Pascal Bergstr\"{a}{\ss}er and Anthony Lin are supported by the
\grantsponsor{erc}{ERC}{https://erc.europa.eu}
Starting Grant
\grantnum{erc}{759969} (AV-SMP).
\end{acks}

\label{beforebibliography}
\newoutputstream{pages}
\openoutputfile{main.pg}{pages}
\addtostream{pages}{\getpagerefnumber{beforebibliography}}
\closeoutputstream{pages}

\bibliographystyle{ACM-Reference-Format}
\bibliography{references.bib}

\appendix
\section{Directed vs undirected cliques}\label{sec:app-undirected}
\input{app-undirected}

\section{Constructions of Büchi automata}\label{sec:app-buchi}
\input{app-buchi}

\section{Proof of Proposition~\ref{prop:clique-rec}}\label{sec:clique-rec}
\input{app-equivalence}

\section{Recurrent reachability with generalized Büchi condition}
\label{sec:app-generalized}

\input{app-generalized}

\section{Proofs of Section~\ref{sec:unranked}}\label{sec:app-unranked}
\input{app-unranked}

\label{endofdocument}
\newoutputstream{pagestotal}
\openoutputfile{main.pgt}{pagestotal}
\addtostream{pagestotal}{\getpagerefnumber{endofdocument}}
\closeoutputstream{pagestotal}

\end{document}

%% file: macros.tex
\newcommand{\ialphabet}{\Sigma}

\newcommand{\OMIT}[1]{}

\newcommand{\sem}[2]{[\![#1 ]\!]_{#2}}

\renewcommand{\A}{\mathcal{A}}
\newcommand{\B}{\mathcal{B}}
\renewcommand{\C}{\mathcal{C}}
\renewcommand{\D}{\mathcal{D}}

\newcommand{\N}{\mathbb{N}}
\newcommand{\T}{\mathcal{T}}
\newcommand{\U}{\mathcal{U}}
\newcommand{\F}{\mathcal{F}}

\newcommand{\rec}{\mathit{Rec}}
\newcommand{\ram}{\exists^{\mathrm{ram}}}
\newcommand{\kram}[1]{\exists^{#1\text{-}\mathrm{ram}}}

\newcommand{\enc}{\mathrm{enc}}
\newcommand{\Enc}{\mathrm{Enc}}
\newcommand{\fcns}{\mathrm{fcns}}

\newcommand{\rewrite}[1]{\stackrel{#1}{\hookrightarrow}}

\newcommand{\rk}{\mathrm{rk}}

\newcommand{\dom}{\mathrm{dom}}
\newcommand{\nodes}{\mathrm{nodes}}
\newcommand{\holes}{\mathrm{holes}}

\newcommand{\svector}[2]{\begin{psmallmatrix} #1 \\ #2 \end{psmallmatrix}}

\newcommand{\pre}{<_{\mathsf{p}}}
\newcommand{\pref}{\le_{\mathsf{p}}}

\newclass{\TwoEXP}{2EXP}

%% file: abstract.tex
Automatic structures are infinite structures that are finitely represented by 
synchronized finite-state automata. This paper concerns specifically automatic 
structures over finite words and trees (ranked/unranked).
We investigate the ``directed version'' of Ramsey quantifiers, which
express the existence of an infinite  directed clique. This subsumes the standard
``undirected version'' of Ramsey quantifiers.
Interesting connections between Ramsey quantifiers and two problems in verification
are firstly observed:
(1) reachability with B\"{u}chi and generalized B\"{u}chi conditions in regular 
model checking can be seen as Ramsey quantification over transitive automatic
graphs (i.e., whose edge relations are transitive), (2) checking monadic 
decomposability (a.k.a. recognizability) of automatic relations can be 
viewed as Ramsey quantification
over co-transitive automatic graphs (i.e., the complements of whose edge 
relations are transitive). We provide a comprehensive complexity
landscape of Ramsey quantifiers in these three cases (general, transitive,
co-transitive), all between \NL~and \EXP. In turn, this yields a 
wealth of new results with precise complexity, e.g., %
verification of subtree/flat prefix rewriting, as well as monadic 
decomposability over tree-automatic relations. %
We also obtain substantially simpler proofs, e.g., for \NL~complexity for monadic
decomposability over word-automatic relations (given by DFAs).

%% file: intro_sec.tex
Automatic structures are infinite structures whose domains are \emph{regular
sets} (represented by finite automata over finite words/trees) and whose
relations are \emph{regular relations} (represented by synchronized finite
word/tree automata)~\cite{BG00,BG04}.  They have been intensively studied in
the logic and automata community, and have been also applied to infinite-state
verification, especially the area of regular model checking.
In this paper, we are interested in algorithmic aspects of the 
\emph{Ramsey quantifiers} \cite{Rubin08}, which state the existence of 
infinite cliques. [Ramsey quantifiers are also known in model theory as
Magidor-Malitz quantifiers, e.g., see \cite{Hayut17}.]
Similar clique-like objects were also considered in the proofs of two seemingly
different problems over automatic structures; namely, the problems of
verification of liveness properties in regular model checking \cite{TL08,TL10}, 
as well as monadic decomposability of regular relations (a.k.a. finite 
recognizability) \cite{barcelo2019monadic,CartonCG06}.
In this paper, we investigate a more general notion of quantifiers that
generalize the classic Ramsey quantifiers and these aforementioned clique-like 
objects studied in the context of liveness and monadic decomposability.
Through our new notion of quantifiers, we establish a comprehensive picture of 
the computational complexity landscape of Ramsey quantifiers over automatic
structures, as well the aforementioned verification problems.
We first discuss the state-of-the-art of these seemingly disconnected problems.

\paragraph{Ramsey Quantifiers over Automatic Structures}
Blumensath and Gr\"{a}del were the first to systematically study automatic
structures~\cite{BG00,BG04}. A fundamental fact is that, given a first-order
formula $\varphi(\bm{x})$ and a word/tree automatic structure $\mathfrak{A}$
(with domain $A$), one can effectively compute a synchronized word/tree
automaton representing the set $\sem{\varphi}{\mathfrak{A}} = \{ \bm{a} \in
A^{|\bm{x}|} \mid \mathfrak{A} \models \varphi(\bm{a}) \}$ of solutions of
$\varphi$. In other words, regular relations are \emph{effectively closed under
all first-order operations}. Consequently, first-order~(FO) model checking over
automatic structures is decidable. 

In the seminal paper \cite{BG00} on automatic structures, it was already
observed that FO can be enriched with the quantifier ``there exist infinitely
many'' --- i.e., $\exists^{\infty} x \colon \varphi(x,\bm{y})$, which is true iff there
exist infinitely many $x$ such that $\varphi(x,\bm{y})$ --- while preserving
the above effective closure property and decidability of model checking. In
fact, assuming nondeterministic automata as finite representations of
$\sem{\varphi}{\mathfrak{A}}$, one can compute
$\sem{\exists^{\infty} x \colon \varphi}{\mathfrak{A}}$ in
polynomial-time.  A few years later, in the case of word automatic structures, Rubin~\cite{Rubin08} studied Ramsey quantifiers, which generalize $\exists^{\infty}$
by enforcing that these infinitely many elements form an infinite undirected
clique, and showed that Ramsey quantifiers preserve regularity as well, meaning FO
extended with Ramsey quantifiers is still decidable. Upon closer inspection,
Rubin's construction runs in doubly exponential time. \emph{Is this optimal?}
\emph{Does the same extend to tree-automatic structures?}

\paragraph{Liveness in Regular Model Checking}
\emph{Regular model checking (RMC)} is a generic verification framework that exploits 
regular languages and relations (e.g., over finite/$\omega$-words 
or trees) as symbolic representations of infinite systems
\cite{tree-rmc,rmc-survey,LR21}.
Various
flavors of automata and transducers for representing the transition relations 
are employed, e.g., word/tree automatic relations (or subsets thereof), 
$\omega$-automatic relations (or subsets thereof), and rational relations.
Since safety and reachability are undecidable in RMC (e.g.,
over automatic graphs), one focus of RMC has been to develop 
acceleration/widening techniques, which are semi-algorithms for computing 
reachability sets (i.e.,  $post^*( S )$) and reachability relations (i.e., 
transitive closure $R^*$ of the edge relation in the graph), that may
terminate on many interesting cases. Some of these semi-algorithms have general
completeness and termination guarantee, e.g., bounded local-depth acceleration
for automatic relations are guaranteed to compute reachability relations for 
pushdown systems (PDS) and ground-tree 
rewrite systems (GTRS) \cite{tree-rmc,tree-rmc2}, while flattable acceleration
for Presburger-definable relations is guaranteed to compute reachability 
relations for reversal-bounded counter systems and 2-dimensional vector 
addition systems with states \cite{LS06}. 

Reachability sets/relations can be directly used to solve safety. The challenge 
of verifying liveness is the necessity to deal with genuinely 
infinite paths (with no repeated configurations). To and Libkin~\cite{TL08,TL10} showed that one can decide
liveness (in the form of recurrent reachability) over word- and tree-automatic 
graphs, when the transitive closure $R^*$ of the edge relation is additionally 
supplied (e.g., by the aforementioned acceleration methods or otherwise). The 
algorithm runs in time polynomial in the size of the problem, with $R^*$ 
supplied as part of inputs. Their technique uses a kind of ``staircase
argument'' combined with Ramsey's Theorem to construct a B\"{u}chi automaton that
represents some witnessing infinite runs.

\paragraph{Monadic Decomposability and Recognizability}
A classic task in the theory of finite-state transductions is the 
problem to decide whether a given regular relation $R \subseteq
(\ialphabet^*)^k$ is \emph{recognizable}, i.e., if it can be expressed as a
finite union of cartesian products of regular languages (in symbols:
$R = \bigcup_{i=1}^n L_{i,1} \times \cdots \times L_{i,k}$ for some $n \in \N$ and
regular sets $L_{i,j} \subseteq \ialphabet^*$). In the formal verification 
terminology~\cite{monadic-decomposition}, such a relation is said to be 
\emph{monadically decomposable}, i.e., that it can be expressed as a Boolean 
combination of monadic predicates. The first important result was by Stearns~\cite{Stearns67} and Valiant~\cite{Valiant75}: Their algorithms for checking regularity of deterministic pushdown automata imply that given
deterministic rational relation---i.e., a relation 
$R \subseteq \ialphabet^* \times
\ialphabet^*$ recognized by \emph{deterministic asynchronous automata}, which is strictly more
general than binary automatic relations---can be checked to be recognizable in
doubly exponential time. This decidability was extended to general $k$-ary
deterministic rational relations by Carton et al.~\cite{CartonCG06}, which
yields decidability as well for the subclass of automatic relations. As noted by
Löding and Spinrath \cite{LS19}, the complexity of the algorithm for automatic relations in
\cite{CartonCG06} runs in doubly exponential time. Using their new 
polynomial-time algorithm for checking regularity for deterministic visibly
pushdown automata, Löding and Spinrath showed that this could be improved
to single exponential time for binary automatic relations. The complexity for
automatic relations was fully settled by Barcel\'o et al.
\cite{barcelo2019monadic} by showing that this problem is \NL-complete (resp.
\PSPACE-complete) when $R$ is presented as a deterministic (resp.
nondeterministic) automaton. The proof technique in~\cite{barcelo2019monadic} is an extremely intricate refinement and analysis of the staircase argument used
by To and Libkin~\cite{TL08} for recurrent reachability for automatic relations.

\paragraph{Contributions} Our account of the state-of-the-art of the
aforementioned three research directions seems to suggest that there might be
some connections between them. To what extent are they connected? 
Is there a more fundamental notion that unifies them? 
These questions are hitherto 
open, but as we shall see in this paper the answer is a resounding yes. 
We pinpoint that the \emph{directed Ramsey
quantifiers}---which ask for the existence of infinite directed cliques (instead of infinite
undirected cliques as in \cite{Rubin08})---is a fundamental concept that
underlies the above three problems, and lets us study them under the same
umbrella, while inferring the optimal complexity and even new results. On the 
one hand, the directed Ramsey quantifiers subsume
the standard Ramsey quantifiers. 

On the other hand,
recurrent reachability over automatic graphs~\cite{TL08} can be
seen as a Ramsey quantifier over a transitive binary relation, whereas monadic
decomposability over automatic relations~\cite{CartonCG06,barcelo2019monadic,LS19} 
can be construed as a Ramsey quantifier over co-transitive binary relations. 
Our results are summarized in \Cref{fig:results}.

Firstly, from the proof by Barcel\'o et al. \cite{barcelo2019monadic},
it is possible to infer that the Ramsey quantifier can be evaluated on regular 
relations in $\NL$,
which substantially improves the doubly exponential-time algorithm of Rubin
\cite{Rubin08}. Unfortunately, their argument relies on an intricate Ramsey 
argument on the transition monoid of the automaton.
Our contribution is a substantially simpler argument that avoids the use of
the transition monoid altogether, which we show to generalize to the case of 
tree-regular relations (which is not the case with the proof of
\cite{barcelo2019monadic}).
More precisely, our approach divides the proof for regular relations into two 
steps:
(i) First, we argue that one can assume infinite cliques witnessed by accepting runs that form
a {\em comb of combs}.
(ii) Then, we argue that the runs can be ``merged'' together
so that it can be witnessed by a single run of a polynomial-size Büchi
automaton. This way, we obtain the same complexity as~\cite{barcelo2019monadic}.

For tree-regular relations we can easily extend step (i).
The comb of combs structure of the accepting runs can be witnessed by an alternating Büchi tree automaton,
which yields the complexity of $\EXP$ for the Ramsey quantifier on tree-regular relations.
However, step (ii) is provably impossible over tree-regular relations, since as we show, the infinite clique problem is $\EXP$-hard.
For the special cases of transitive and co-transitive relations we need further separate arguments
that enable us to evaluate the Ramsey quantifier in $\P$. The case for
transitive relations can be inferred from the proof in \cite{TL08}, but not so
for the co-transitive case.

Finally, we apply our results to decidability and complexity of 
recurrent reachability with generalized Büchi conditions, and automatic
structures over unranked trees. We show, for example, decidability (in fact
in polynomial-time) of recurrent reachability of subtree/flat prefix rewriting,
answering an open question by Löding and Spelten~\cite{loding2007transition}
and decidability (in fact, $\EXP$-completeness) of recurrent reachability with 
generalized Büchi conditions of ground tree rewrite systems, answering an open
question by Löding~\cite{Loding06}.

\paragraph{Organization} We provide a more detailed summary of our main results
in \cref{sec:results}. We fix notation and basic terminologies in
\cref{sec:preliminaries}. We then start with the word case in
\cref{sec:words} and proceed to the tree case in \cref{sec:trees}.
Applications and generalizations to unranked trees are given in, respectively,
\cref{sec:generalized} and \cref{sec:unranked}.

\begin{table*}

\small
\renewcommand\tabularxcolumn[1]{m{#1}}
\newcolumntype{C}{>{\centering\arraybackslash}X}
\renewcommand{\arraystretch}{1.2}

\begin{tabularx}{0.9\textwidth}{>{\hsize=0.28\hsize}C >{\hsize=0.3\hsize}C >{\hsize=0.42\hsize}C}
 & regular relations & tree-regular relations \\\toprule
\makecell{Automaton construction \\ for Ramsey quantifier} & logspace & \makecell{exponential time \\ poly-time for transitive relations} \\\midrule
\makecell{Recurrent reachability \\ \& infinite clique } & \NL-complete$^*$ & \makecell{\EXP-complete for NTA, D$\downarrow$TA \\ 
\P-complete for transitive$^*$ or \\ co-transitive relations or D$\uparrow$TA} \\\midrule
Recurrent reachability with generalized Büchi condition & \PSPACE-complete & \EXP-complete \\\midrule
Monadic decomposability & \makecell{\NL-complete for DFA$^*$ \\ \PSPACE-complete for NFA$^*$} & \makecell{\P-complete for D$\uparrow$TA, D$\downarrow$TA \\ \EXP-complete for NTA} \\\bottomrule
\end{tabularx}
\caption{Complexity results. Those marked with $^*$ were known, but we provide simpler proofs. The other results are new.
}
\label{fig:results}
\end{table*}

%% file: results.tex
To improve readability, we provide a detailed summary of our main results 
in this section before we take a deeper dive into the proofs. 
Unless otherwise specified, the completeness results mentioned in this section
(and \cref{fig:results}) hold for 
NFAs and DFAs in the word case and
NTAs, D$\uparrow$TAs, and D$\downarrow$TAs in the tree case.
We define the directed Ramsey quantifier:

\begin{definition}
Let $\mathfrak{A}$ be a structure with domain $A$.
The \emph{Ramsey quantifier} $\ram$
over an $\mathfrak{A}$-formula $\varphi$ with $k+2$ free variables
is defined for all $\bm{c} \in A^k$ by
$\mathfrak{A} \models \ram x,y \colon \varphi(x,y,\bm{c})$ if and only if
there is an infinite sequence $(a_i)_{i \ge 1}$ of pairwise distinct elements $a_i \in A$
so that $\mathfrak{A} \models \varphi(a_i,a_j,\bm{c})$ for all $1 \le i < j$.
\end{definition}
We deviate from the definition of the Ramsey quantifier found in the literature, see \cite{gradel2020automatic},
requiring $\mathfrak{A} \models \varphi(a_i,a_j,\bm{c})$ {\em for all} $i \neq j$, in the definition above.
Over (tree-)regular relations the two quantifier definitions can be simulated by each other, see \Cref{sec:app-undirected}.
Furthermore, there are also higher-dimensional versions $\kram{d}$ of the Ramsey quantifier,
which will not be considered in this paper.

\paragraph{Evaluating Ramsey quantifiers}
If $R$ is a binary \mbox{(tree-)} regular relation, then evaluating $\ram x,y \colon R(x,y)$ is the problem
of checking whether $R$ contains an \emph{infinite (directed) clique}, i.e., an infinite sequence
$(a_i)_{i \ge 1}$ of distinct elements of $A$ such that $(a_i,a_j) \in R$ for all $1 \le i < j$.
It follows from \cite{barcelo2019monadic} that the infinite clique problem over word-regular relations is $\NL$-complete.
We provide a much simpler proof by considering a slightly more general setting.
Instead of the infinite clique problem we consider the evaluation of the Ramsey quantifier on a $(k+2)$-ary
(tree-)regular relation $R \subseteq A^{k+2}$, i.e., compute an automaton for
$\sem{\ram x,y \colon R(x,y,\bm{z})}{} = \{ \bm{c} \in A^k \mid \ram x,y \colon R(x,y,\bm{c}) \}$.

\begin{restatable}{theorem}{wordramsey}
	\label{thm:word-ramsey}
	Given a regular relation $R \subseteq (\Sigma^*)^{k+2}$ by an NFA $\A$\footnotemark
	\footnotetext{In this and the following theorems, the parameter $k$ is part of the input.},
	one can construct an NFA for the relation $\sem{\ram x,y \colon R(x,y,\bm{z})}{}$ in logspace.
	In particular, the infinite clique problem over regular relations is in $\NL$.
\end{restatable}

We show that the complexity of the infinite clique problem increases from $\NL$ to $\EXP$ 
when considered over tree-regular relations given by NTAs or D$\downarrow$TAs.
Let $\T_\Sigma$ denote the set of ranked trees over alphabet $\Sigma$.
\begin{restatable}{theorem}{treecomplete}
\label{thm:tree-lowerbound}
The infinite clique problem over tree-regular relations $R \subseteq \T_\Sigma \times \T_\Sigma$ is 
$\EXP$-complete if $R$ is given as NTA or D$\downarrow$TA, and
$\P$-complete if $R$ is given as D$\uparrow$TA.
\end{restatable}
For the exponential lower bound, we present a reduction from intersection nonemptiness for NTAs and D$\downarrow$TAs.  
This is surprising, because an analogue reduction in the word case does not
exist: This would yield a $\PSPACE$ lower bound for the infinite clique problem
over words, but the latter belongs to $\NL$.

For the exponential upper bound of \cref{thm:tree-lowerbound}, we prove the tree analogue of \cref{thm:word-ramsey}.
It even holds when the relation $R$ is given as an alternating tree automaton (ATA),
which allows us to apply it to recurrent reachability with {\em generalized Büchi condition}.

\begin{restatable}{theorem}{treeramsey}\label{thm:tree-ramsey}
Given an ATA (D$\uparrow$TA) $\A$ for a tree-regular relation $R \subseteq (\T_\Sigma)^{k+2}$,
one can construct in exponential (polynomial) time an NTA for the relation $\sem{\ram x,y \colon R(x,y,\bm{z})}{}$.
\end{restatable}

If we make further assumptions on the relation $R$, we obtain a better complexity for NTAs.
We say that a $(k+2)$-ary relation $R$ over $A$ is \emph{transitive}
if the binary relation $\{(a,b) \mid (a,b,\bm{c}) \in R\}$ is transitive for all $\bm{c} \in A^k$.
\begin{restatable}{theorem}{transitiveramsey}\label{thm:transitive-ramsey}
Given an NTA $\A$ for a transitive tree-regular $R \subseteq (\T_\Sigma)^{k+2}$, one can construct in polynomial time an NTA for the relation $\sem{\ram x,y \colon R(x,y,\bm{z})}{}$.
In particular, the infinite clique problem over transitive tree-regular relations is in $\P$.
\end{restatable}

A binary relation $R \subseteq A \times A$ is {\em co-transitive} if its complement $(A \times A) \setminus R$ is a transitive relation.

\begin{restatable}{theorem}{cotransitive}
\label{thm:cotransitive}
The infinite clique problem over co-transitive tree-regular relations
$R \subseteq \mathcal{T}_\Sigma \times \mathcal{T}_\Sigma$ given as NTA is $\P$-complete.
\end{restatable}

In \cref{sec:unranked} we show by a reduction that the Ramsey quantifier can be evaluated
over unranked tree-regular relations with the same complexity as in the ranked case.

\paragraph{Recurrent reachability}

Since reachability in automatic graphs is in general undecidable \cite{BG04},
we will instead use {\em transitive paths}, i.e., infinite sequences $(a_i)_{i \ge 1}$ with $(a_i,a_j) \in R$
for all $1 \le i < j$.
Given sets $L_1, \dots, L_k \subseteq A$ we write $\rec(L_1, \dots, L_k)[R]$ for the set of all initial vertices $a_1$
of transitive paths $(a_i)_{i \ge 1}$ that visit each set $L_j$ infinitely often.
{\em Recurrent reachability with generalized Büchi condition} is the problem of testing
$a_1 \in \rec(L_1, \dots, L_k)[R]$ for a given (tree-)regular relation $R \subseteq A \times A$,
(tree-)regular languages $L_1, \dots, L_k \subseteq A$, and an initial element $a_1 \in A$.
If $k = 1$ this problem is simply called {\em recurrent reachability}.

Since the infinite clique problem and recurrent reachability are logspace
equivalent (\cref{prop:clique-rec}), we obtain: 
\begin{corollary}\label{cor:rec-reach}
Recurrent reachability is $\NL$-complete over regular relations.
It is $\EXP$-complete over tree-regular relations given by NTAs or D$\downarrow$TAs, and
$\P$-complete if the tree-regular relations are transitive or given by D$\uparrow$TAs.
\end{corollary}

We also apply \cref{thm:word-ramsey,thm:tree-ramsey} to obtain tight upper
bounds for recurrent reachability with generalized Büchi condition. The lower
bounds result from a reduction from intersection nonemptiness.
\begin{restatable}{theorem}{generalbuchi}\label{thm:general-buchi}
Recurrent reachability with generalized Büchi condition is $\PSPACE$-complete over regular relations,
and $\EXP$-complete over tree-regular relations.
\end{restatable}

\paragraph{Monadic Decomposability}

Recall that a relation $R \subseteq A^k$ is \emph{monadically decomposable} if it is of the form
$\bigcup_{i=1}^n A_{i,1} \times \dots \times A_{i,k}$
for some $n \in \N$ and \mbox{(tree-)regular} languages $A_{i,j}$.
The traditional approach to deciding monadic decomposability~\cite{ginsburg1966bounded,libkin2003variable,CartonCG06,LS19,barcelo2019monadic} is to associate with $R$
certain equivalence relations $\sim_j$ for $1\le j\le k$ such that $R$ is monadically decomposable if and only if each $\sim_j$ has finite index. 
An equivalence relation has infinite index if and only if there exist infinitely many elements that are pairwise in different equivalence classes
which is witnessed by an infinite clique in the complement relation.
Therefore, monadic decomposability amounts to checking that $\sim_j$'s complement $\not\sim_j$ does not have an infinite clique for any $j$. 
If $R$ is given by a DFA (resp.\ NFA), then one can construct an NFA for each $\not\sim_j$ in logspace (resp.\ in $\PSPACE$) and thus \cref{thm:word-ramsey} yields a tight upper bound:
\begin{restatable}{corollary}{wordrecognizable}\label{cor:word-recognizable}
Given a regular relation $R \subseteq (\Sigma^*)^k$ by a DFA (resp. NFA),
it is $\NL$-complete (resp. $\PSPACE$-complete) to decide whether $R$ is monadically decomposable.
\end{restatable}
While this approach yields optimal complexity for words, this is, unexpectedly, not the case for trees.
For a tree-regular relation given as D$\downarrow$TA or D$\uparrow$TA (resp.\ NTA), one can
also construct an NTA for each $\not\sim_j$ in logspace (resp.\ $\PSPACE$).
Then, applying \cref{thm:tree-ramsey} would yield an $\EXP$ (resp.\ $\TwoEXP$)
algorithm. However, perhaps surprisingly, monadic decomposability for trees has much lower complexity:
\begin{restatable}{corollary}{recognizable}\label{cor:recognizable}
Given a tree-regular relation $R \subseteq \T_\Sigma^k$ by a D$\downarrow$TA or D$\uparrow$TA (resp. NTA),
it is $\P$-complete (resp. $\EXP$-complete) to decide whether $R$ is monadically decomposable.
\end{restatable}
To get the $\P$ (resp.\ $\EXP$) algorithm, we exploit the co-transitivity of each $\not\sim_j$ and apply \cref{thm:cotransitive} instead of \cref{thm:tree-ramsey}.
This shows the importance of the co-transitivity notion: In the word case, monadic decomposability requires only the generic clique detection, but the tree case is more nuanced---we need one algorithm for the general case and a specialized algorithm for co-transitive relations.

%% file: preliminaries.tex
We assume familiarity with the basic models of (non)deter\-ministic and alternating finite automata on words and trees
as well as with standard complexity classes (e.g., \NL, \P, \PSPACE, \EXP).
We refer the reader to the textbooks \cite{Kozen,tata} for more details.
We often abbreviate a finite or infinite sequence of elements $a_1, a_2, \dots$
by a boldface letter $\bm{a}$. 

\paragraph{Trees}
A {\em tree domain} is a nonempty set $D \subseteq \N^*$ such that
(i) $D$ is prefix closed, i.e., $uv \in D$ implies $u \in D$,
(ii) for all $v \in \N^*$ and $j \le i$ if $vi \in D$, then $vj \in D$, and
(iii) each node $v \in D$ has only finitely many {\em children} $vi \in D$ where $i \in \N$.
An {\em unranked tree} over an alphabet $\Sigma$ is a function $t \colon \dom(t) \to \Sigma$
where $\dom(t)$ is a finite tree domain.
A {\em ranked alphabet} is a finite alphabet $\Sigma$ where every symbol $a \in \Sigma$
has a rank $\rk(a) \in \N$.
A {\em ranked tree} is an unranked tree $t$ such that
every node $v \in \dom(t)$ has $\rk(t(v))$ many children.
We denote the set of all ranked and unranked trees over $\Sigma$ by $\T_\Sigma$ and $\U_\Sigma$, respectively.

Let $x \notin \Sigma$ be a {\em variable}.
The set $\C_\Sigma$ of all {\em contexts} over $\Sigma$
contains all unranked trees over $\Sigma \cup \{x\}$
such that every node $u \in \dom(t)$ with $t(u) = x$ is a leaf, called \emph{hole}.
We partition $\dom(t) = \nodes(t) \cup \holes(t)$ into nodes and holes.
The size of a context $t$ is $|\nodes(t)|$.
For contexts $s, t_1, \dots, t_n$ with $|\holes(s)| = n$ we denote by $s[t_1,\dots,t_n]$
the context obtained by replacing the $i$-th hole in lexicographic order by $t_i$.
For two contexts $s_1,s_2 \in \C_\Sigma$, we call $s_1$ a {\em prefix} of $s_2$, denoted by $s_1 \pref s_2$,
if $s_1[t_1, \dots, t_n] = s_2$ for some contexts $t_1, \dots, t_n$.
If $n >0$ and each context $t_i$ has size at least one, then $s_1$ is a {\em proper prefix} of $s_2$, denoted by $s_1 \pre s_2$.

We define an \emph{infinite unranked tree} and an \emph{infinite ranked tree} as in the finite case but with infinite domains.
We denote the set of all finite and infinite unranked trees over the alphabet $\Sigma$ by $\U_\Sigma^\infty$ and
the set of all finite and infinite ranked trees over $\Sigma$ by $\T_\Sigma^\infty$.

\paragraph{Regular and tree-regular languages}
A \emph{nondeterministic finite automaton} (NFA) over the alphabet $\Sigma$ is a tuple $\A = (Q,\Sigma,\Delta,q_0,F)$ where $Q$ is a finite set of states, $\Delta \subseteq Q \times (\Sigma \cup \{\varepsilon\}) \times Q$ is a transition relation, $q_0 \in Q$ is an initial state, and $F \subseteq Q$ is a set of final states.
We denote by $L(\A) \subseteq \Sigma^*$ the {\em regular language} recognized by $\A$. 
In our algorithms, the alphabet $\Sigma$ is \emph{not} part of the representation of an automaton. 
Instead, we will always work with the subalphabet of all symbols occurring in the transitions. 
This will be important later when the implicitly given alphabet is significantly smaller.

A \emph{nondeterministic (top-down) tree automaton} (NTA) over the ranked alphabet $\Sigma$ is a tuple $\A = (Q,\Sigma,\Delta,q_0)$ where $Q$ is a finite set of states, $q_0 \in Q$ is an initial state, and $\Delta \subseteq \bigcup_{a \in \Sigma} Q \times \{a\} \times Q^{\rk(a)}$ is a transition relation.
A \emph{run} of $\A$ on a tree $t \in \T_\Sigma$ is a tree $\rho \in \T_Q$ with $\dom(\rho) = \dom(t)$ such that
$\rho(\varepsilon) = q_0$ and
$(\rho(u),t(u),\rho(u1),\dots,\rho(ur)) \in \Delta$ for all nodes $u \in \dom(\rho)$ with $\rk(t(u)) = r$.
As before, $L(\A)$ is the set of trees recognized by $\A$, i.e., the set of all trees $t$
such that there exists a run of $\A$ on $t$.
A set of trees is called \emph{tree-regular} if there is an NTA that recognizes it.
We will also use the notions of \emph{deterministic finite automata} (DFA),
\emph{deterministic bottom-up} (D$\uparrow$TA), and \emph{deterministic top-down tree automata} (D$\downarrow$TA).
Moreover, \emph{alternating automata} will be formally introduced in later sections.

\paragraph{Regular and tree-regular relations}
Let $\Sigma$ be a finite alphabet and let $\Sigma_\bot = \Sigma \cup \{\bot\}$ where $\bot \notin \Sigma$ is a fresh symbol.
For words $w_1, \dots, w_k \in \Sigma^*$ with $w_i = a_{i,1} \dots a_{i,n_i}$ and $n := \max\{n_i \mid 1 \le i \le k\}$ we define their convolution
\[
	w_1 \otimes \dots \otimes w_k := \begin{bmatrix} w_1 \\ \vdots \\ w_k \end{bmatrix} := \begin{pmatrix} a_{1,1}' \\ \vdots \\ a_{k,1}' \end{pmatrix} \dots \begin{pmatrix} a_{1,n}' \\ \vdots \\ a_{k,n}' \end{pmatrix} \in (\Sigma_\bot^k)^*
\]
where $a_{i,j}' = a_{i,j}$ if $j \le n_i$ and $a_{i,j}' = \bot$ otherwise.
A relation $R \subseteq (\Sigma^*)^k$ is {\em recognized} by an NFA $\A$ if 
$L(A) = \{ w_1 \otimes \dots \otimes w_k \mid (w_1, \dots, w_k) \in R \}$.
In that case we call $R$ {\em regular}.

We extend the definitions to tree-regular relations.
For an alphabet $\Sigma$ we set again $\Sigma_\bot = \Sigma \cup \{\bot\}$ where $\bot \notin \Sigma$ is a fresh symbol.
Let $\varepsilon$ be the \emph{empty tree} with $\dom(\varepsilon) := \emptyset$.
Given $k$ trees $t_1, \dots, t_k \in \U_\Sigma \cup \{\varepsilon\}$ we define their convolution
$t = t_1 \otimes \dots \otimes t_k \in \U_{\Sigma^k_\bot} \cup \{\varepsilon\}$ with $\dom(t) = \bigcup_{i=1}^k \dom(t_i)$ and $t(v) = (t'_1(v), \dots, t'_k(v))$ 
where $t'_i(v) = t_i(v)$ if $v \in \dom(t_i)$ and $t'_i(v) = \bot$ otherwise.
Observe that the degree of a node $v$ in $t_1 \otimes \dots \otimes t_k$ is the maximum degree of $v$ in a tree $t_i$ such that $v \in \dom(t_i)$.
Similar to the word case, we also write the convolution of trees as a column vector.
If all $t_i$ are ranked trees, then also $t$ is a ranked tree with $\rk(a_1,\dots,a_k) := \max\{\rk(a_i) \mid 1 \le i \le k\}$ for all $(a_1,\dots,a_k) \in \Sigma_\bot^k$ where $\rk(\bot) := 0$.
A relation $R \subseteq \T_\Sigma^k$ is {\em recognized} by an NTA $\A$
if the tree language $\{ t_1 \otimes \dots \otimes t_k \mid (t_1, \dots, t_k) \in R \}$ is recognized by $\A$.
In that case we call $R$ {\em tree-regular}.

Regular and tree-regular relations are effectively closed under first-order operations
(Boolean operations and projections).
A relational structure $\mathfrak{A}$ is {\em automatic (tree-automatic)} if its universe and all its relations are regular (tree-regular).

%% file: word.tex
\tikzstyle{nt}=[color=black!80, fill=black, fill opacity=0.05]
\tikzstyle{context}=[color=blue!80, fill=blue, fill opacity=0.1]
\tikzstyle{tree}=[color=red!80, fill=red, fill opacity=0.3]
\tikzstyle{contextfade}=[color=blue!80, top color=blue, bottom color=white, fill=white, fill opacity=0.1]
\tikzstyle{treefade}=[color=red!80, top color=red, bottom color=white, fill=white, fill opacity=0.3]
\tikzstyle{par}=[circle,draw=none,fill=black,inner sep=1.5pt, fill opacity=1]

\newcommand{\ra}{\texttt{\textcolor{red}{a}}}
\newcommand{\rb}{\texttt{\textcolor{red}{b}}}
\newcommand{\ba}{\texttt{\textcolor{blue}{a}}}
\newcommand{\bb}{\texttt{\textcolor{blue}{b}}}
\newcommand{\obb}{\textcolor{blue}{\overline{\texttt{b}}}}
\newcommand{\bo}{{\bot}}

\newcommand{\mra}{\makebox[7pt]{\textcolor{red}{$\mathtt{a}$}}}
\newcommand{\mrb}{\makebox[7pt]{\textcolor{red}{$\mathtt{b}$}}}
\newcommand{\mba}{\makebox[7pt]{\textcolor{blue}{$\mathtt{a}$}}}
\newcommand{\mbb}{\makebox[7pt]{\textcolor{blue}{$\mathtt{b}$}}}
\newcommand{\mbo}{\makebox[7pt]{\bot}}

\begin{figure}
\centering
\begin{tikzpicture}[xscale=0.9,yscale=0.9,anchor=base, baseline]

\tikzstyle{p}=[fill, circle, inner sep = 1pt]

\node[p] (0) at (0,0) {};
\node[p] (1) at (1,0) {};
\node[p] (2) at (2,0) {};
\node[p] (3) at (3,0) {};
\node[p] (4) at (4,0) {};
\node[p] (5) at (5,0) {};
\node[p] (6) at (6,0) {};
\node[p] (7) at (7,0) {};
\node[p] (8) at (8,0) {};
\node[p] (9) at (9,0) {};

\node[p] (1') at (1,0.5) {};
\node[p] (2') at (2,1) {};
\node[p] (4') at (4,0.5) {};
\node[p] (5') at (5,1) {};
\node[p] (6') at (6,0.5) {};
\node[p] (7') at (7,1) {};

\draw (0) edge node [yshift=-.8em] {$\ba$} (1);
\draw (1) edge node [yshift=-.8em] {$\bb$} (2);
\draw (2) edge node [yshift=-.8em] {$\ba$} (3);
\draw (3) edge node [yshift=-.8em] {$\ba$} (4);
\draw (4) edge node [yshift=-.8em] {$\ba$} (5);
\draw (5) edge node [yshift=-.8em] {$\bb$} (6);
\draw (6) edge node [yshift=-.8em] {$\ba$} (7);
\draw (7) edge node [yshift=-.8em] {$\bb$} (8);
\draw (8) edge node [yshift=-.8em] {$\ba$} (9);

\draw (0) edge node [above] {$\rb$} (1');
\draw (1') edge node [above] {$\ra$} (2');

\draw (3) edge node [above] {$\rb$} (4');
\draw (4') edge node [above] {$\rb$} (5');

\draw (5) edge node [above] {$\ra$} (6');
\draw (6') edge node [above] {$\rb$} (7');

\end{tikzpicture}

\vspace{1em}

$
\begin{bmatrix} \rb \\ \ba \end{bmatrix}
\begin{bmatrix} \ra \\ \bb \end{bmatrix}
\begin{bmatrix} \bo \\ \ba \end{bmatrix}
\#
\begin{bmatrix} \rb \\ \ba \end{bmatrix}
\begin{bmatrix} \rb \\ \ba \end{bmatrix}
\#
\begin{bmatrix} \ra \\ \bb \end{bmatrix}
\begin{bmatrix} \rb \\ \ba \end{bmatrix}
\begin{bmatrix} \bo \\ \bb \end{bmatrix}
\begin{bmatrix} \bo \\ \ba \end{bmatrix}
\#
\dots
$

\caption{An example word comb $\texttt{ba}, \texttt{ababb}, \texttt{abaaaab}, \dots$ and its encoding as an infinite word.}
\label{fig:comb}
\end{figure}
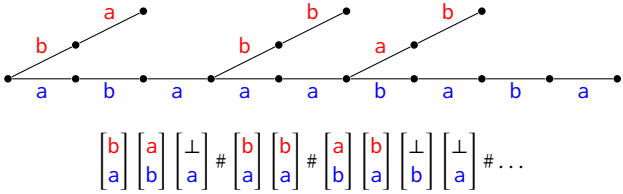

We first consider Ramsey quantifiers over word-regular relations.
We show that if $R \subseteq (\Sigma^*)^{k+2}$ is a regular relation, 
then an automaton for
$\{ \bm{c} \mid \ram x,y \colon R(x,y,\bm{c}) \}$
can be constructed in logarithmic space (\cref{thm:word-ramsey}).

\paragraph{Word combs}
The first step is to observe that, when looking for infinite cliques in $R$,
one can restrict to combs:
An infinite sequence $\bm{v}$ of words is called a {\em comb}
if there exist infinite sequences $\bm{\alpha}$ and $\bm{\beta}$ of words
with $v_i = \beta_1 \dots \beta_{i-1} \alpha_i$
and $1 \le |\alpha_i| \le |\beta_i|$ for all $i \ge 1$.
The pair $(\bm{\alpha},\bm{\beta})$ is called a \emph{generator} of $\bm{v}$.
We remark that the choice of the generator is not unique.
Any infinite subsequence of a comb is again a comb.
In fact, the following lemma is well-known, see \cite[Lemma~5.1]{KuskeL10}.

\begin{lemma}
	\label{lem:comb}
	Any sequence $\bm{w}$ of pairwise distinct words $w_i$ over a finite alphabet $\Sigma$ contains a comb as a subsequence. 
\end{lemma}

In contrast to arbitrary infinite sequences of words,
combs can be encoded naturally by infinite words.
If $(\bm{\alpha},\bm{\beta})$ is a generator we call the infinite word
\[
	\enc(\bm{\alpha},\bm{\beta}) = \begin{bmatrix} \alpha_1 \\ \beta_1 \end{bmatrix} \# \begin{bmatrix} \alpha_2 \\ \beta_2 \end{bmatrix} \# \cdots \in ((\Sigma_\bot \times \Sigma) \cup \{\#\})^\omega
\]
the {\em encoding of $(\bm{\alpha},\bm{\beta})$}, or also an encoding of $\bm{v}$.

\paragraph{Comb of combs}
The next goal would be to construct a Büchi automaton which reads an encoding
of a comb $\bm{v}$ and verifies that $v_i \otimes v_j$ has an accepting run $\rho(v_i,v_j)$ for all $1 \le i < j$.
In general, this is challenging since it is not clear how a finite automaton can keep track of infinitely many runs
(let alone, an automaton of polynomial size).
Instead we will show that every infinite clique contains an infinite subclique
whose accepting runs can be arranged in a dag of constant width,
and can therefore be recognized by a polynomial-sized Büchi automaton.

Consider a comb $\bm{v}$ with generator $(\bm{\alpha}, \bm{\beta})$.
First observe that the convolutions $v_i \otimes v_j$ can be written as
\begin{equation}\label{eq:decomposition}
	\begin{bmatrix}v_i \\ v_j\end{bmatrix} = \begin{bmatrix}\beta_1 \\ \beta_1\end{bmatrix} \dots \begin{bmatrix}\beta_{i-1} \\ \beta_{i-1}\end{bmatrix} \begin{bmatrix}\alpha_i \\ \beta_i\end{bmatrix} 
	\begin{bmatrix}\varepsilon \\ \beta_{i+1}\end{bmatrix} \dots \begin{bmatrix}\varepsilon \\ \beta_{j-1}\end{bmatrix} \begin{bmatrix}\varepsilon \\ \alpha_j\end{bmatrix}
\end{equation} 
and can hence be arranged in a trie displayed in \Cref{fig:comb-of-combs}, that we call {\em comb of combs}.
The next insight is that we can ensure that the accepting runs $\rho(v_i,v_j)$ on the convolutions $v_i \otimes v_j$
match this comb of combs structure, after replacing $\bm{v}$ by an infinite subsequence.
Roughly speaking, the runs look like as if the automaton for $R$ would be {\em deterministic}.
For example, all runs $\rho(v_1,v_j)$ for $j > 1$ share a common prefix which is a run on $\alpha_1 \otimes \beta_1$,
and all runs $\rho(v_i,v_j)$ for $1 < i < j$ share a common prefix which is a run on $\beta_1 \otimes \beta_1$.

A \emph{run} of an NFA $\A = (Q,\Sigma,\Delta,q_{\mathrm{in}},F)$ on a nonempty word $a_1 \dots a_n \in \Sigma^*$ (runs on the empty word are not needed here) is
a sequence $(q_0,a_1,q_1) (q_1,a_2,q_2) \dots (q_{n-1},a_n,q_n)$ of triples in $Q \times \Sigma \times Q$
such that there exists a path from $q_{i-1}$ to $q_i$ in $\A$ labeled with $a_i$ for all $1 \le i \le n$.
A run is \emph{accepting} if $q_0 = q_{\mathrm{in}}$ and $q_n \in F$.

We define a \emph{decomposition} of a word $w \in \Sigma^*$ as $w = u_1 \dots u_n$ where $u_i \in \Sigma^*$.
The decomposition in \cref{eq:decomposition} is called the $(\bm{\alpha},\bm{\beta})$-decomposition of $v_i \otimes v_j$.
We say that a decomposition of a run $\rho = \rho_1 \dots \rho_n$ of an NFA is \emph{compatible} with a decomposition $w = u_1 \dots u_n$ of a word if 
$\rho_i$ is a run on $u_i$ for all $i \in [1,n]$.

We say that a generator $(\bm{\alpha},\bm{\beta})$ of a comb $\bm{v}$ is \emph{coarser} than a generator $(\bm{\gamma},\bm{\delta})$ of a comb $\bm{w}$
if there exist indices $k_1 < k_2 < \dots$ such that
$v_i = w_{k_i}$ and
$\beta_1 \dots \beta_i = \delta_1 \dots \delta_{k_i}$
for all $i \ge 1$.
In this case we also say that $(\bm{\alpha},\bm{\beta})$ is the coarsening of $(\bm{\gamma},\bm{\delta})$ defined by the subsequence $\bm{v}$ of $\bm{w}$.

\begin{lemma}\label{lem:comb-of-combs}
Let $\bm{w}$ be a comb generated by $(\bm{\gamma},\bm{\delta})$ that forms an infinite clique in a regular relation $R \subseteq \Sigma^* \times \Sigma^*$ given as an NFA $\A$. 
There exist a coarsening $(\bm{\alpha},\bm{\beta})$ of $(\bm{\gamma},\bm{\delta})$ that generates a comb $\bm{v}$, 
accepting runs $\rho(v_i,v_j)$ of $\A$ on $v_i \otimes v_j$, and runs $\kappa_i, \lambda_i, \mu_{i,j}, \nu_{i,j}$ such that
\[\rho(v_i,v_j) = \kappa_1 \dots \kappa_{i-1} \lambda_i \mu_{i,i+1} \dots \mu_{i,j-1} \nu_{i,j}\]
is a decomposition compatible with
the $(\bm{\alpha},\bm{\beta})$-decomposition of $v_i \otimes v_j$
for all $i < j$.
\end{lemma}

\begin{proof}
Let $\bm{w}$ be a comb generated by $(\bm{\alpha},\bm{\beta})$ that forms an infinite clique in $R$ and 
$\rho(w_i,w_j)$ be an accepting run of $\A$ on $w_i \otimes w_j$ for all $1 \le i < j$.
We establish the run structure as illustrated in \cref{fig:buchi-simulation} column-wise.

Assume we already defined $(\kappa_i)_{i < n}$, $(\lambda_i)_{i < n}$, $(\mu_{i,j})_{i < j < n}$, and $(\nu_{i,j})_{i < j < n}$ for some $n \ge 1$ such that
\begin{align*}
\rho(w_i,w_j) &= \kappa_1 \dots \kappa_{i-1} \lambda_i \mu_{i,i+1} \dots \mu_{i,n-1} \tau(w_i,w_j)\\
\rho(w_{i'},w_{j'}) &= \kappa_1 \dots \kappa_{n-1} \sigma(w_{i'},w_{j'})
\end{align*}
for runs $\tau(w_i,w_j)$, $\sigma(w_{i'},w_{j'})$ for all $1 \le i \le n-1 < j$ and $n-1 < i' < j'$.
For all $1 \le i < n$ define $\nu_{i,n} := \tau(w_i,w_n)$.

We now define $\mu_{i,n}$ successively for each $1 \le i < n$.
In step $i$ we apply the pigeonhole principle to get an infinite subsequence $\bm{v}$ of $\bm{w}$ starting with $w_1,\dots,w_n$
such that all runs $\tau(v_i,v_j)$ for $j > n$ have a common prefix 
$\mu_{i,n}$ which is a run on $\varepsilon \otimes \beta_n$.
At the end of step $i$ we replace $\bm{w}$ by $\bm{v}$ and we replace $(\bm{\alpha},\bm{\beta})$ by the coarsening defined by $\bm{v}$.

Next we define $\lambda_n$.
By the pigeonhole principle there exists an infinite subsequence $\bm{v}$ of $\bm{w}$ starting with $w_1,\dots,w_n$
such that all runs $\sigma(v_n,v_j)$ for $j > n$ have a common prefix
$\lambda_n$ which is a run on $\alpha_n \otimes \beta_n$.
Again we replace $\bm{w}$ by $\bm{v}$ and $(\bm{\alpha},\bm{\beta})$ by the coarsening defined by $\bm{v}$.

Finally, by Ramsey's theorem there is an infinite subsequence $\bm{v}$ of $\bm{w}$ starting with $w_1,\dots,w_n$ 
such that all runs $\sigma(v_i,v_j)$ for $n < i < j$ have a common prefix
$\kappa_n$ which is a run on $\beta_n \otimes \beta_n$.
We replace $\bm{w}$ by $\bm{v}$ and $(\bm{\alpha},\bm{\beta})$ by the coarsening defined by $\bm{v}$.

In the limit we obtain the desired decomposition of the runs $\rho(v_i,v_j)$ and the generator $(\bm{\alpha},\bm{\beta})$ of a comb $\bm{v}$ that is coarser than the initial generator of $\bm{w}$.
\end{proof}

It is not hard to see that such a comb of combs structure can be simulated by an alternating Büchi automaton,
which would only yield a $\PSPACE$-solution for the infinite clique problem.
The following key lemma states that the runs $\mu_{i,j}$, $\nu_{i,j}$ from \Cref{lem:comb-of-combs}
can be chosen independently from $i$,
which reduces the width of the run dag of the alternating automaton to a constant.

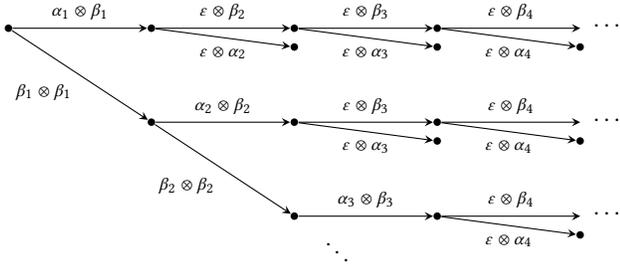
\begin{figure}

\centering

\begin{tikzpicture}[xscale=1.9,yscale=1.25]

\tikzstyle{p}=[fill, circle, inner sep = 1pt]
\tikzstyle{i}=[draw = none, minimum size = 0, inner sep = 0, outer sep = 0]
\tikzstyle{e}=[>=stealth,->]

\node[p] (00) at (0,0) {};
\node[p] (01) at (1,0) {};
\node[p] (11) at (1,-1) {};
\node[p] (12) at (2,-1) {};
\node[p] (22) at (2,-2) {};
\node[p] (23) at (3,-2) {};
\node at (2.3,-2.3) {$\ddots$};

\draw[e] (00) edge node [above] {\scriptsize $\alpha_1 \otimes \beta_1$} (01);
\draw[e] (00) edge node [below left] {\scriptsize $\beta_1 \otimes \beta_1$} (11);
\draw[e] (11) edge node [above] {\scriptsize $\alpha_2 \otimes \beta_2$} (12);
\draw[e] (11) edge node [below left] {\scriptsize $\beta_2 \otimes \beta_2$} (22);
\draw[e] (22) edge node [above] {\scriptsize $\alpha_3 \otimes \beta_3$} (23);

\node[p] (02) at (2,0) {};
\node[p] (03) at (3,0) {};
\node[i] (04) at (4,0) {};
\node at (4.2,0) {$\cdots$};
\node[p] (01a) at (2,-0.2) {};
\node[p] (02a) at (3,-0.2) {};
\node[p] (03a) at (4,-0.2) {};

\draw[e] (01) edge node [above] {\scriptsize $\varepsilon \otimes \beta_2$} (02);
\draw[e] (02) edge node [above] {\scriptsize $\varepsilon \otimes \beta_3$} (03);
\draw[e] (03) edge node [above] {\scriptsize $\varepsilon \otimes \beta_4$} (04);
\draw[e] (01) edge node [below] {\scriptsize $\varepsilon \otimes \alpha_2$} (01a);
\draw[e] (02) edge node [below] {\scriptsize $\varepsilon \otimes \alpha_3$} (02a);
\draw[e] (03) edge node [below] {\scriptsize $\varepsilon \otimes \alpha_4$} (03a);

\node[p] (13) at (3,-1) {};
\node[i] (14) at (4,-1) {};
\node at (4.2,-1) {$\cdots$};
\node[i] (15) at (5,-1) {};
\node[p] (12a) at (3,-1.2) {};
\node[p] (13a) at (4,-1.2) {};

\draw[e] (12) edge node [above] {\scriptsize $\varepsilon \otimes \beta_3$} (13);
\draw[e] (13) edge node [above] {\scriptsize $\varepsilon \otimes \beta_4$} (14);
\draw[e] (12) edge node [below] {\scriptsize $\varepsilon \otimes \alpha_3$} (12a);
\draw[e] (13) edge node [below] {\scriptsize $\varepsilon \otimes \alpha_4$} (13a);

\node[i] (24) at (4,-2) {};
\node at (4.2,-2) {$\cdots$};
\node[i] (25) at (5,-2) {};
\node[p] (23a) at (4,-2.2) {};

\draw[e] (23) edge node [above] {\scriptsize $\varepsilon \otimes \beta_4$} (24);
\draw[e] (23) edge node [below] {\scriptsize $\varepsilon \otimes \alpha_4$} (23a);

\end{tikzpicture}

\caption{If $\bm{v}$ is a comb of the form $v_i = \beta_1 \dots \beta_{i-1} \alpha_i$
then the convolutions $v_i \otimes v_j$ for all $i < j$ form a {\em comb of combs}.}
\label{fig:comb-of-combs}

\end{figure}

\begin{lemma}\label{lem:easy-combs}
If $\bm{w}$ is an infinite clique in a regular relation $R \subseteq \Sigma^* \times \Sigma^*$ given as an NFA $\A$, then there exist
a generator $(\bm{\alpha},\bm{\beta})$ for a subsequence $\bm{v}$ of $\bm{w}$,
accepting runs $\rho(v_i,v_j)$ of $\A$ on $v_i \otimes v_j$, and runs $\kappa_i, \lambda_i, \mu_{j}, \nu_{j}$ such that
\[\rho(v_i,v_j) = \kappa_1 \dots \kappa_{i-1} \lambda_i \mu_{i+1} \dots \mu_{j-1} \nu_{j}\]
is a decomposition compatible with
the $(\bm{\alpha},\bm{\beta})$-decomposition of $v_i \otimes v_j$
for all $i < j$.
\end{lemma}

\begin{proof}
Suppose that $R$ has an infinite clique.
Then there exist an infinite clique $\bm{w}$ in $R$ generated by $(\bm{\alpha},\bm{\beta})$ and
runs $\kappa_i, \lambda_i, \mu_{i,j}, \nu_{i,j}$ for $i < j$ as in \Cref{lem:comb-of-combs}.

It remains to ensure that $\mu_{i,j} = \mu_{i',j}$ and $\nu_{i,j} = \nu_{i',j}$ for all $i < i' < j$.
To do so, consider the initial state of the run $\mu_{i,j}$.
By Ramsey's theorem there exist indices $k_1 < k_2 < \dots$ such that
all runs $\mu_{k_i,k_j}$ have the same initial state.
We define $v_i = w_{k_i}$ for all $i \ge 1$, and
\begin{align*}
	\tilde \kappa_1 &= \kappa_1 \cdots \kappa_{k_2-1}
	& \tilde \lambda_1 &= \kappa_1 \cdots \kappa_{k_1-1} \lambda_{k_1} \mu_{k_1,k_1+1} \cdots \mu_{k_1,k_2-1}\\
	\tilde \kappa_i &= \kappa_{k_i} \cdots \kappa_{k_{i+1}-1} 
	& \tilde \lambda_i &= \lambda_{k_i} \mu_{k_i,k_i+1} \cdots \mu_{k_i,k_{i+1}-1} \\
	\tilde \mu_i &= \mu_{k_1,k_i} \cdots \mu_{k_1,k_{i+1}-1}
	& \tilde \nu_i &= \nu_{k_1,k_i}
\end{align*}
for all $i \ge 2$.
Observe that the composition $\tilde \lambda_i \tilde \mu_{i+1}$ forms a valid run
since $\mu_{k_i,k_{i+1}}$ and $\mu_{k_1,k_{i+1}}$ have the same initial state.
Then $\tilde \kappa_1 \dots \tilde \kappa_{i-1} \tilde \lambda_i \tilde \mu_{i+1} \dots \tilde \mu_{j-1} \tilde \nu_j$ is an accepting run
on $v_i \otimes v_j$.
Furthermore, this run decomposition is compatible with
the $(\bm{\delta},\bm{\gamma})$-decomposition of $v_i \otimes v_j$
where the generator $(\bm{\delta}, \bm{\gamma})$ is defined as
\begin{align*}
\delta_1 &= \beta_1 \cdots \beta_{k_2-1}
& \gamma_1 &= \beta_1 \cdots \beta_{k_1-1} \alpha_{k_1}\\
\delta_i &= \beta_{k_i} \cdots \beta_{k_{i+1}-1}
&\gamma_i &= \alpha_{k_i}
\end{align*}
for all $i \ge 2$.
This concludes the proof.
\end{proof}

A \emph{nondeterministic Büchi automaton} (NBA)
has the same format as an NFA $\B = (Q,\Sigma,\Delta,q_0,F)$.
An infinite word $w \in \Sigma^\omega$ is {\em accepted} by $\B$
if there exists an accepting run $(q_0,a_1,q_1) (q_1,a_2,q_2) \dots \in \Delta^\omega$
with $w = a_1 a_2 \dots$ and $q_i \in F$ for infinitely many $i \ge 0$.

\begin{restatable}{proposition}{combsbuechi}
	\label{prop:combs-buechi}
	Given an NFA $\A$ for a relation $R \subseteq (\Sigma^*)^2$,
	one can construct in logarithmic space a Büchi automaton $\B$ over the alphabet $(\Sigma_\bot \times \Sigma) \cup \{\#\}$
	such that:
	\begin{itemize}
	\item If $\bm{w}$ is an infinite clique in $R$
	then $\B$ accepts an encoding of a comb $\bm{v}$ which is a subsequence of $\bm{w}$.
	\item If $\B$ accepts an encoding of a comb $\bm{w}$
	then $\bm{w}$ is an infinite clique in $R$.
	\end{itemize}
\end{restatable}
\begin{proof}
Given an NFA $\A = (Q,\Sigma_\bot^2,q_{\mathrm{in}},\Delta,F)$, we add to $\A$ a fresh state $\bot$
and transitions $\bot \xrightarrow{(a,b)} q$ for all $(a,b) \in \Sigma_\bot^2$ and $q \in Q_\bot$ where $Q_\bot := Q \cup \{\bot\}$.

The Büchi automaton $\B = (Q_\bot^4,\Sigma_\bot^2,q_{\mathrm{in}}^\B,\Delta^\B,Q_\bot^4)$
simulates the runs $\kappa_j,\lambda_j,\mu_j,\nu_j$ from \Cref{lem:easy-combs}
in the four components.
Its initial state is $q_{\mathrm{in}}^\B = (q_{\mathrm{in}},q_{\mathrm{in}},\bot,\bot)$
and it contains the following transtions:
\begin{itemize}
\item $(p,s,q,t) \xrightarrow{(a,b)}_{\B} (p',s',q',t')$ for all transitions $p \xrightarrow{(b,b)}_{\A} p'$, $s \xrightarrow{(a,b)}_{\A} s'$,
	$q \xrightarrow{(\bot,b)}_{\A} q'$, $t \xrightarrow{(\bot,a)}_{\A} t'$,
\item $(p,q,q,t) \xrightarrow{\#}_{\B} (p,p,q,q)$ for all $p,q \in Q$, $t \in F$,
\item $(p_1,p_2,p_3,p_4) \xrightarrow{\varepsilon}_\B (q_1,q_2,q_3,q_4)$
if $p_i = q_i$ or $p_i \xrightarrow{\varepsilon}_\A q_i$ for all $i \in [1,4]$.
\end{itemize}
The desired Büchi automaton is a product automaton of $\B$
and a Büchi automaton which verifies that the input word is a valid comb encoding $\enc(\bm{\alpha},\bm{\beta})$.
\end{proof}

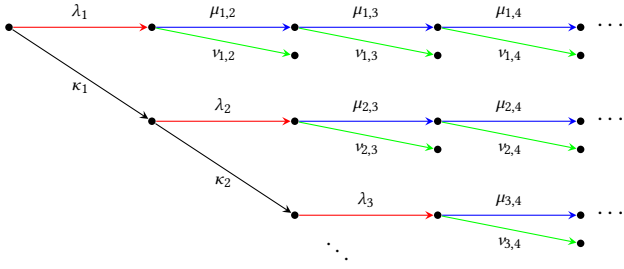
\begin{figure}
\centering

\begin{tikzpicture}[xscale=1.9,yscale=1.25]

\tikzstyle{p}=[fill, circle, inner sep = 1pt]
\tikzstyle{i}=[draw = none, minimum size = 0, inner sep = 0, outer sep = 0]
\tikzstyle{e}=[>=stealth,->]

\node[p] (00) at (0,0) {};
\node[p] (01) at (1,0) {};
\node[p] (11) at (1,-1) {};
\node[p] (12) at (2,-1) {};
\node[p] (22) at (2,-2) {};
\node[p] (23) at (3,-2) {};
\node at (2.3,-2.3) {$\ddots$};

\draw[e, red] (00) edge node[above, black]{\scriptsize $\lambda_1$} (01);
\draw[e] (00) edge node[below]{\scriptsize $\kappa_1$} (11);
\draw[e, red] (11) edge node[above, black]{\scriptsize $\lambda_2$} (12);
\draw[e] (11) edge node[below]{\scriptsize $\kappa_2$} (22);
\draw[e, red] (22) edge node[above, black]{\scriptsize $\lambda_3$} (23);

\node[p] (02) at (2,0) {};
\node[p] (03) at (3,0) {};
\node[p] (04) at (4,0) {};
\node at (4.2,0) {$~\cdots$};
\node[p] (01a) at (2,-0.3) {};
\node[p] (02a) at (3,-0.3) {};
\node[p] (03a) at (4,-0.3) {};

\draw[e, blue] (01) edge node[above, black]{\scriptsize $\mu_{1,2}$} (02);
\draw[e, blue] (02) edge node[above, black]{\scriptsize $\mu_{1,3}$} (03);
\draw[e, blue] (03) edge node[above, black]{\scriptsize $\mu_{1,4}$} (04);
\draw[e, green] (01) edge node[below, black]{\scriptsize $\nu_{1,2}$} (01a);
\draw[e, green] (02) edge node[below, black]{\scriptsize $\nu_{1,3}$} (02a);
\draw[e, green] (03) edge node[below, black]{\scriptsize $\nu_{1,4}$} (03a);

\node[p] (13) at (3,-1) {};
\node[p] (14) at (4,-1) {};
\node at (4.2,-1) {$~\cdots$};
\node[p] (12a) at (3,-1.3) {};
\node[p] (13a) at (4,-1.3) {};

\draw[e, blue] (12) edge node[above, black]{\scriptsize $\mu_{2,3}$} (13);
\draw[e, blue] (13) edge node[above, black]{\scriptsize $\mu_{2,4}$} (14);
\draw[e, green] (12) edge node[below, black]{\scriptsize $\nu_{2,3}$} (12a);
\draw[e, green] (13) edge node[below, black]{\scriptsize $\nu_{2,4}$} (13a);

\node[p] (24) at (4,-2) {};
\node at (4.2,-2) {$~\cdots$};
\node[p] (23a) at (4,-2.3) {};

\draw[e, blue] (23) edge node[above, black]{\scriptsize $\mu_{3,4}$} (24);
\draw[e, green] (23) edge node[below, black]{\scriptsize $\nu_{3,4}$} (23a);

\end{tikzpicture}

\caption{One can always find an infinite comb clique $\bm{v}$ with accepting runs
which can be decomposed in the form
$\rho(v_i,v_j) = \kappa_1 \dots \kappa_{i-1} \lambda_i \mu_{i,i+1} \dots \mu_{i,j-1} \nu_{i,j}$.}
\label{fig:buchi-simulation}
\end{figure}

\begin{proof}[Proof of~\Cref{thm:word-ramsey}]
First observe that we can construct in log-space an NFA $\C$ over $\Sigma$ such that
(i) for every infinite clique $\bm{w}$ of $R$ some element $w_i$ is accepted by $\C$, and
(ii) if $w$ is accepted by $\C$ then $w$ belongs to an infinite clique of $R$.
To be more precise, $\C$ accepts $\alpha_1 \in \Sigma^*$ if and only if
some encoding $\enc(\bm{\alpha},\bm{\beta})$ is accepted by
the Büchi automaton $\B$ from \Cref{prop:combs-buechi}.
This can be done in log-space as follows:
First we construct a Büchi automaton $\hat \C$ over $\Sigma \cup \{\#\}$
which accepts all words of the form $\alpha_1 \#^\omega$
such that an encoding $\enc(\bm{\alpha},\bm{\beta})$ is accepted by $\B$.
Then $\hat \C$ is turned into an NFA $\C$ (which does not read the suffix $\#^\omega$)
by replacing $\#$-transitions by $\varepsilon$-transitions.
Furthermore $\C$ tracks the number of final states visited so far and accepts if and only if this number exceeds
the number of states in $\hat \C$.

Given an NFA $\A$ for $R \subseteq (\Sigma^*)^{k+2}$.
We first construct an NFA $\A'$ over $\Sigma_\bot^{2k+2}$
which accepts the regular binary relation
\[
	R' = \{ (u \otimes c_1 \otimes \dots \otimes c_k, v \otimes c_1 \otimes \dots \otimes c_k) \mid (u,v,\bm{c}) \in R \}.
\]
Note that $\A'$ can be constructed in logspace since it is obtained by taking each transition of $\A$ and duplicating the $\bm{c}$-coordinates and moving the $v$-coordinate.
Let $\C'$ be the NFA described above which accepts at least one word from each infinite $R'$-clique
and only accepts elements of infinite $R'$-cliques.
Projecting away the first component yields the desired NFA for $\{ \bm{c} \in (\Sigma^*)^k \mid \ram x,y \colon R(x,y,\bm{c}) \}$.
\end{proof}

%% file: hardness.tex
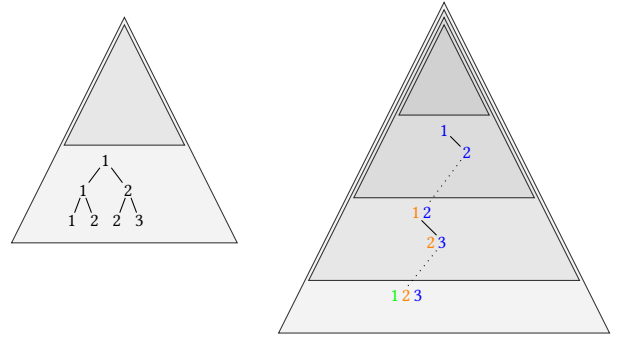
\begin{figure}
\begin{tikzpicture}[yscale=-1]

\begin{scope}[shift={(1.25,0.1)}]
\filldraw[nt] (0,0) -- (0.8,1.6) -- (-0.8,1.6) -- cycle;
\end{scope}

\begin{scope}[shift={(1.25,0)}]
\filldraw[nt] (0,0) -- (1.5,3) -- (-1.5,3) -- cycle;
\end{scope}

\tikzstyle{aut}=[inner sep = 0.5pt]%
\begin{scope}[shift={(1,1.9)}]
\node[aut] (a) at (0,0) {\footnotesize 1};
\node[aut] (b) at (-0.3,0.4) {\footnotesize 1};
\node[aut] (c) at (0.3,0.4) {\footnotesize 2};
\node[aut] (d) at (-0.45,0.8) {\footnotesize 1};
\node[aut] (e) at (-0.15,0.8) {\footnotesize 2};
\node[aut] (f) at (0.15,0.8) {\footnotesize 2};
\node[aut] (g) at (0.45,0.8) {\footnotesize 3};

\draw (a) -- (b);
\draw (a) -- (c);
\draw (b) -- (d);
\draw (b) -- (e);
\draw (c) -- (f);
\draw (c) -- (g);
\end{scope}

\begin{scope}[shift={(5.5,0.1)}]
\filldraw[nt] (0,0) -- (0.6,1.2) -- (-0.6,1.2) -- cycle;
\end{scope}

\begin{scope}[shift={(5.5,0)}]
\filldraw[nt] (0,0) -- (1.2,2.4) -- (-1.2,2.4) -- cycle;
\end{scope}

\begin{scope}[shift={(5.5,-0.1)}]
\filldraw[nt] (0,0) -- (1.8,3.6) -- (-1.8,3.6) -- cycle;
\end{scope}

\begin{scope}[shift={(5.5,-0.2)}]
\filldraw[nt] (0,0) -- (2.2,4.4) -- (-2.2,4.4) -- cycle;
\end{scope}

\begin{scope}[shift={(5.5,1.5)}]
\node[aut] (a) at (0,0) {\footnotesize \textcolor{blue}{1}};
\node[aut] (b) at (0.3,0.3) {\footnotesize \textcolor{blue}{2}};
\node[aut] (c) at (-0.3,1.1) {\footnotesize \textcolor{orange}{1}\,\textcolor{blue}{2}};
\node[aut] (d) at (-0.1,1.5) {\footnotesize \textcolor{orange}{2}\,\textcolor{blue}{3}};
\node[aut] (e) at (-0.5,2.2) {\footnotesize \textcolor{green}{1}\,\textcolor{orange}{2}\,\textcolor{blue}{3}};

\draw (a) -- (b);
\draw[dotted] (b) -- (c);
\draw (c.south) -- (d.north);
\draw[dotted] (d.south) -- (e.north);
\end{scope}

\end{tikzpicture}

\caption{Left: Illustration of the tree automaton $\mathcal{T}$, tracking the number of right directions modulo $n$.
Right: A path on which the runs of $\mathcal{T}$ are disjoint.}
\label{fig:lower-bound}

\end{figure}

In this section we prove the exponential lower bound from \cref{thm:tree-lowerbound} for the infinite clique problem over trees.
This lower bound is surprising since over words, the infinite clique problem
can be reduced to the emptiness of (word) Büchi automata, which is $\NL$-complete.
This is not the case in the tree case since emptiness of Büchi tree automata is $\P$-complete.

We start with an intutive explanation of the lower bound.
To prove the upper bound in the word case,
we used the fact that we can assume cliques $\bm{v}$ whose runs $\rho(v_i,v_j)$ can be merged into a \emph{single} global run
(\cref{lem:easy-combs}).
Over tree regular relations this is not the case anymore.
Consider a deterministic top-down tree automaton~$\mathcal{T}$,
which behaves as follows on the convolution $t \otimes t'$ of two binary trees $t, t'$
with $\dom(t) \subsetneq \dom(t')$:
Starting from every node on the fringe of $t$, the automaton tracks the number of times it moves to a right child, modulo some number $n$;
see \Cref{fig:lower-bound} for a depiction.
Now consider an increasing sequence of binary trees $t_1, t_2, \dots$, and the unique runs $\rho_{i,j}$ of $\mathcal{T}$ on $t_i \otimes t_j$.
\Cref{fig:lower-bound} illustrates that we can always find a path on which the runs $\rho_{1,n}, \rho_{2,n}, \dots, \rho_{n-1,n}$ are disjoint.
This behavior indicates that it is difficult to witness the existence of infinite cliques by a polynomially-sized Büchi tree automaton.

We extend this idea to a reduction from the intersection non-emptiness problem for tree automata,
which is known to be $\EXP$-complete \cite{FruhwirthSVY91}:
Given an NTA $\A = (Q,\Sigma,\Delta,q_0)$ and states $q_1, \dots, q_n \in Q$, decide whether $\bigcap_{i=1}^n L(\A_{q_i})$ is non-empty.
Here, $\A_{q_i}$ denotes the NTA $\A$ with initial state $q_i$.

We construct a relation $R$ on decorated trees, which are obtained from a binary tree by attaching
to every inner node $u$ a ranked tree $\delta(u)$ over $\Sigma$. 
Let $\Gamma := \Sigma \cup \{a,c\}$ be a ranked alphabet with $\rk(a) = 3,~\rk(c) = 0$.
A {\em decorated tree} is a tree $t \in \T_\Gamma$ such that $t(\varepsilon) \in \{a,c\}$
and for all $u \in \dom(t)$ we have
\begin{itemize}
\item $t(u1) = t(u2) \in \{a,c\}$ and $t(u3) \in \Sigma$ if $t(u) = a$ and
\item $t(ui) \in \Sigma$ for all $i \in [1,\rk(t(u))]$ if $t(u) \in \Sigma$.
\end{itemize}
We denote by $a(t) := \{u \in \dom(t) \mid t(u) = a\}$ the nodes of $t$ labeled with $a$ and by $ac(t) := \{u \in \dom(t) \mid t(u) \in \{a,c\}\}$ the nodes labeled with $a$ or $c$.
The \emph{decoration} of $t$ is a function $\delta_t \colon a(t) \to \T_\Sigma$ such that $\delta_t(u) = t_{\downarrow u3}$ where $t_{\downarrow v}$ denotes the subtree of $t$ rooted in $v \in \dom(t)$.

Let $\A' = (Q',\Gamma,\Delta',p_1)$ be the NTA where $Q' := Q \cup \{p_1, \dots, p_n\}$
and $\Delta'$ contains all transitions from $\Delta$ and the transitions
\begin{align*}
p_i & \xrightarrow{a} (p_i,p_{i+1},q_i) \text{ for all } i \in [1,n] \text{ where } p_{n+1} := p_1, \\
p_i & \xrightarrow{c} () \text{ for all } i \in [1,n].
\end{align*}
We define the tree-regular relation $R \subseteq \T_\Gamma \times \T_\Gamma$ such that $(s,t) \in R$ if and only if
$s$ and $t$ are decorated trees with $ac(s) \subseteq a(t)$ and 
$\A'$ accepts $t_{\downarrow u}$ for all $u \in \min(a(t) \setminus a(s))$.
Here, the minimum is defined with respect to prefix ordering.
It is easy to construct an NTA that recognizes $R$ in logspace.

It remains to show that $\bigcap_{i=1}^n L(\A_{q_i}) \neq \emptyset$ if and only if $R$ contains an infinite clique.
For the ``only if'' direction let $t \in \T_\Sigma$ be a tree that is accepted by $\A_{q_i}$ for all $i \in [1,n]$.
For all $i \ge 0$ we define the decorated tree $t_i \in \T_\Gamma$ such that $ac(t_i) = \bigcup_{j=0}^i \{1,2\}^j$
and $\delta_{t_i}(u) = t$ for all $u \in a(t_i)$.
It is easy to verify that $(t_i,t_j) \in R$ for all $i < j$.

Conversely, consider a sequence of decorated trees $t_i \in \T_\Gamma$ for $i \ge 0$ with $(t_i,t_j) \in R$ for all $i < j$.
We define nodes $v_1, \dots, v_n$ with
\begin{itemize}
\item $v_i \in \min(a(t_i) \setminus a(t_{i-1}))$ for all $i \in [1,n]$ and
\item $v_{i+1} = v_i 2 1^{k_i}$ for all $i \in [1,n-1]$ and some $k_i \ge 0$.
\end{itemize}
We can choose $v_1 \in \min(a(t_1) \setminus a(t_0))$ arbitrary which defines $v_2, \dots, v_n$ uniquely.
Since $(t_{i-1},t_n) \in R$ and $v_i \in \min(a(t_n) \setminus a(t_{i-1}))$, the subtree ${t_n}_{\downarrow v_i}$ is accepted by $\A'$ for all $i \in [1,n]$.
By definition of $\A'$ there exist accepting runs on ${t_n}_{\downarrow v_n}$ starting from $p_1, \dots, p_n$.
Therefore, the tree $\delta_{t_n}(v_n) \in \T_\Sigma$ is accepted by $\A$ starting from all states $q_1, \dots, q_n$.

We note that $\EXP$-hardness already holds if $R$ is given by a D$\downarrow$TA
since intersection nonemptiness is $\EXP$-hard already for D$\downarrow$TAs \cite{Seidl94},
and if the automaton $\A$ is a D$\downarrow$TA, then the constructed relation $R$ from the proof
can also be recognized by a D$\downarrow$TA.
Moreover, the reduction can be adapted to recurrent reachability by setting the target set to $\T_\Sigma$,
which proves the exponential lower bound in \cref{cor:rec-reach}.

%% file: combs.tex
To prove the upper bounds for tree-regular relations, we extend the notion of combs to the tree case.
Here a tree $t_i$ is decomposed vertically in the form $t_i = \beta_1 \dots \beta_{i-1} \alpha_i$
where $\beta_1$ is a context, $\beta_2, \dots, \beta_{i-1}$ are forests of contexts,
and $\alpha_i$ is a forest of trees, see \Cref{fig:tree-comb} for an abstract illustration.

A \emph{context forest} of width $n$
is a finite sequence $\tau = (c_i)_{1 \le i \le n}$ of contexts $c_i \in \C_\Sigma$.
Context forests of width 1 are regarded as contexts.
We say that $\tau$ is {\em nontrivial} if $n \ge 1$ and $|c_i| \ge 1$ for all $1 \le i \le n$.
If the $c_i$ are trees in $\U_\Sigma$, we call $\tau$ just a \emph{forest}.
We define the concatenation of a context $\tau_1 \in \C_\Sigma$, where $|\holes(\tau_1)|=n$,
with a context forest $\tau_2 = (c_1, \dots, c_n) \in \F_\Sigma$ of width $n$
by $\tau_1 \tau_2 := \tau_1[c_1,\dots,c_n]$.
We write $\tau_1 \tau_2 \dots \tau_n$ for a context $\tau_1$
and context forests $\tau_2, \dots, \tau_n$ assuming left-associativity.
Here we implicitly assume that the width of $\tau_i$ matches $|\holes(\tau_1 \dots \tau_{i-1})|$.
If $t$ is a tree and $s$ is a context, we write $t \triangleleft s$ if $\dom(t) \cap \holes(s)  = \emptyset$.
For a forest $\alpha = (t_i)_{i \le n}$ and a context forest $\beta = (s_i)_{i \le n}$ we also write $\alpha \triangleleft \beta$
if $t_i \triangleleft s_i$ for all $1 \le i \le n$.

An infinite sequence $\bm{t} = (t_i)_{i \ge 1}$ of ranked trees is called a {\em comb}
if there is a sequence of forests $\bm{\alpha} = (\alpha_i)_{i \ge 1}$ and a sequence of nontrivial context forests $\bm{\beta} = (\beta_i)_{i \ge 1}$ such that for all $i \ge 1$ we have
$t_i = \beta_1 \dots \beta_{i-1} \alpha_i$ and
$\alpha_i \triangleleft \beta_i$.
The pair $(\bm{\alpha},\bm{\beta})$ is called \emph{generator} of the comb.
Since the trees $t_i$ are ranked, also the trees in $\alpha_i$ and the contexts in $\beta_i$ are ranked.

The property $\alpha_i \triangleleft \beta_i$ should be compared to the property $|\alpha_i| \le |\beta_i|$ in word combs.
It ensures that every forest $\alpha_i$ does not touch any context forest $\alpha_j, \beta_j$ for $j \neq i$.

\begin{lemma}[Combs lemma over trees]
	\label{lem:tree-comb}
	Any sequence $\bm{t}$ of pairwise distinct ranked trees $t_i \in \T_\Sigma$ over a finite ranked alphabet $\Sigma$ 
	contains a comb as a subsequence.
\end{lemma}

\begin{proof}
	It suffices to show that for any infinite set $T \subseteq \T_\Sigma$ there exists a comb over $T$.
	Consider the following finitely branching infinite tree
	whose nodes are contexts from $\C_\Sigma$.
	The root is the context $x$. The children of a context $s$ 
	are the contexts of the form $s[t_1, \dots, t_n]$ where each $t_i$ is a context of size one.
	Observe that all trees in $\T_\Sigma$ occur as nodes in the infinite tree.
	The set of all ancestors of trees in $T$ form an infinite subtree,
	which contains an infinite path $s_0 \pre s_1 \pre s_2 \pre \dots$ by K{\H o}nig's Lemma.
	For all $i \ge 1$ there exists a tree $t_i \in T$ which contains $s_{i-1}$ as a prefix.

	Since the minimal level of a hole in $s_i$ is strictly increasing,
	for every $i \ge 1$ there exists a $j \ge i$ with $t_i \triangleleft s_j$.
	Hence one can inductively construct indices $1 = k_1 < k_2 < \dots$ such that
	$t_{k_i+1} \triangleleft s_{k_{i+1}}$ for all $i \ge 1$.
	Then $(t_{k_i + 1})_{i \ge 1}$ is a comb where the generator $(\bm{\alpha},\bm{\beta})$ is defined such that
	$s_{k_{i+1}} = \beta_1 \dots \beta_{i}$ and $t_{k_i + 1} = \beta_1 \dots \beta_{i-1} \alpha_i$ for $i \ge 1$.
	The comb property $\alpha_i \triangleleft \beta_i$ follows from $t_{k_i+1} \triangleleft s_{k_{i+1}}$.
\end{proof}

\begin{figure}

\begin{tikzpicture}[xscale=1.2, yscale=-1]

\begin{scope}[shift={(0,0)}]
\node at (0,-0.5) {$t_1$};
\filldraw[tree] (0,0.02) -- (.3,.8) -- (-.3,.8) -- cycle;
\node[red] at (0,0.6) {$\alpha_1$};
\end{scope}

\begin{scope}[shift={(1.25,0)}]
\node at (0,-0.5) {$t_2$};
\filldraw[context] (0,0) -- (.5,1) -- (-.5,1) -- cycle;
\node[blue] at (0,0.6) {$\beta_1$};
\filldraw[tree] (.5,1.01) -- (.6,1.4) -- (-.6,1.4) -- (-.5,1.01) -- cycle;
\node[red] at (0,1.2) {$\alpha_2$};
\end{scope}

\begin{scope}[shift={(3,0)}]
\node at (0,-0.5) {$t_3$};
\node at (1.2,0.8) {$\dots$};
\filldraw[context] (0,0) -- (.5,1) -- (-.5,1) -- cycle;
\node[blue] at (0,0.6) {$\beta_1$};
\filldraw[context] (.5,1) -- (.8,1.6) -- (-.8,1.6) -- (-.5,1) -- cycle;
\node[blue] at (0,1.3) {$\beta_2$};
\filldraw[tree] (.8,1.61) -- (.85,1.9) -- (-.85,1.9) -- (-.8,1.61) -- cycle;
\node[red] at (0,1.77) {$\alpha_3$};
\end{scope}

\begin{scope}[shift={(5.4,0)}]
\node at (0,-0.5) {encoding};
\filldraw[context] (0,0) -- (.5,1) -- (-.5,1) -- cycle;
\filldraw[context] (.5,1) -- (.8,1.6) -- (-.8,1.6) -- (-.5,1) -- cycle;
\filldraw[context] (.8,1.6) -- (1.05,2.1) -- (-1.05,2.1) -- (-.8,1.6) -- cycle;
\filldraw[context,path fading=south] (1.15,2.3) -- (1.05,2.1) -- (-1.05,2.1) -- (-1.15,2.3);
\filldraw[tree] (0,0.02) -- (.3,.8) -- (-.3,.8) -- cycle;
\filldraw[tree] (.5,1.01) -- (.6,1.4) -- (-.6,1.4) -- (-.5,1.01) -- cycle;
\filldraw[tree] (.8,1.61) -- (.85,1.9) -- (-.85,1.9) -- (-.8,1.61) -- cycle;
\filldraw[tree,path fading=south] (1.1,2.3) -- (1.05,2.11) -- (-1.05,2.11) -- (-1.1,2.3);
\end{scope}

\end{tikzpicture}

\caption{An example tree comb and its encoding as an infinite tree.
In this example the generator satisfies $\dom(\alpha_i) \subseteq \dom(\beta_i)$
whereas general generators only satisfy $\alpha_i \triangleleft \beta_i$.}
\label{fig:tree-comb}

\end{figure}
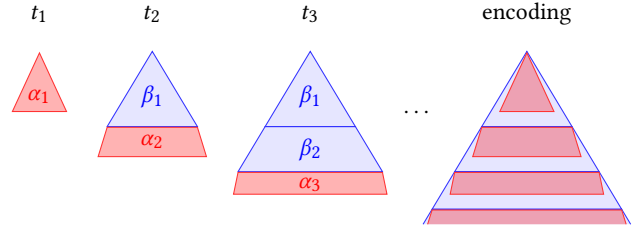

To define the encoding $\enc(\bm{\alpha},\bm{\beta})$ of a comb generator we need a few more definitions.
For a tree $t$ and context $s$ with $t \triangleleft s$ we define the convolution $t \otimes s$ as before but every $(\bot,x)$ is replaced by $x$.
That is, $t \otimes s$ is again a context.
We extend the convolution in a natural way to forests and context forests of the same width.
If $\tau = (c_1, \dots, c_n)$ is a context forest let $\bar \tau$ be obtained from $\tau$
by attaching a new $\#$-labeled root to each of the $n$ contexts $c_i$.
We can now define the \emph{encoding} of a comb with generator $(\bm{\alpha},\bm{\beta})$ as the infinite tree
\[\enc(\bm{\alpha},\bm{\beta}) := (\alpha_1 \otimes \beta_1) (\overline{\alpha_2 \otimes \beta_2})(\overline{\alpha_3 \otimes \beta_3}) \dots\]
over the ranked alphabet $\Omega := \Sigma_\bot^2 \cup \{\#\}$.
See \cref{fig:tree-comb} for an illustration of the encoding.
Here, the forests $\alpha_i$ are colored red and the context forests $\beta_i$ are colored blue.
It is not hard to see that the set $\Enc_\Sigma$ of all comb encodings is a regular language of infinite trees.

%% file: tree.tex
An \emph{alternating tree automaton} (ATA) over the ranked alphabet $\Sigma$ is a tuple $\A = (Q,\Sigma,\delta, q_0)$ where $Q$ is a set of states, $q_0 \in Q$ is an initial state, and $\delta \colon Q \times \Sigma \to \B^+(Q \times \N)$ is a transition function with $\delta(q,a) \in \B^+(Q \times \{1,\dots,\rk(a)\})$ for all $q \in Q,~a \in \Sigma$.
Here, $\B^+(Q \times \N)$ denotes the set of positive propositional formulas over the set of variables $Q \times \N$.
For a set $S$ of variables and formula $\varphi$ we denote by $S \models \varphi$ that if the variables in $S$ are set to true and the variables not in $S$ are set to false, then $\varphi$ is satisfied.

We will use a nonstandard definition of runs of ATAs.
Firstly, we consider runs on both trees and contexts.
Secondly, each node in the run also carries the labels of its children,
with the purpose of predetermining the states in the context holes.
A \emph{run} of $\A$ on a nontrivial ranked context $t \in \C_\Sigma$ is a context $\rho$ over the alphabet
$\N^* \times Q \times 2^{Q \times \N}$ such that 
$\rho(\varepsilon) = (\varepsilon, q_0, S)$ for some $S \subseteq Q \times \N$ and 
for each node $u \in \nodes(\rho)$ with $\rho(u) = (w,q,S)$ and $S = \{(q_1,c_1),\dots,(q_r,c_r)\}$ we have that 
$S \models \delta(q,t(w))$ and $u$ has $r \ge 0$ children such that 
$\rho(ui) = (wc_i, q_i, S_i)$ for some $S_i \subseteq Q \times \N$ if $wc_i \in \nodes(t)$,
and $\rho(ui) = x$, otherwise.
Note that an NTA can be seen as a special ATA where for all $q \in Q$ and $a \in \Sigma$, the transition formula
$\delta(q,a)$ is a disjunction of conjunctions
$\bigwedge_{i=1}^{\rk(a)} (q_i,i)$.

We define a \emph{decomposition} of a context $t$ as $t = \tau_1 \dots \tau_n$ where the $\tau_i$ are context forests.
For a generator $(\bm{\alpha},\bm{\beta})$ of a comb $\bm{s}$ we define the $(\bm{\alpha},\bm{\beta})$-decomposition of $s_i \otimes s_j$ as in the word case.
We say that a decomposition of a run $\rho = \rho_1 \dots \rho_n$ of an ATA is \emph{compatible} with a decomposition $t = \tau_1 \dots \tau_n$ of a context if $\rho_1 \dots \rho_i$ is a run on $\tau_1 \dots \tau_i$ for all $i \in [1,n]$.
Note that the above definition of a run ensures that $\rho_1 \dots \rho_i$ already determines the first two components of the root labels of $\rho_{i+1}$ for all $i < n$.

We say that a generator $(\bm{\alpha},\bm{\beta})$ of a comb $\bm{s}$ is \emph{coarser} than a generator $(\bm{\gamma},\bm{\delta})$ of a comb $\bm{t}$
if there exist indices $k_1 < k_2 < \dots$ such that
$s_i = t_{k_i}$ and
$\beta_1 \dots \beta_i = \delta_1 \dots \delta_{k_i}$
for all $i \ge 1$.
In this case we also say that $(\bm{\alpha},\bm{\beta})$ is the coarsening of $(\bm{\gamma},\bm{\delta})$ defined by the subsequence $\bm{s}$ of $\bm{t}$.

\begin{lemma}\label{lem:tree-run-comb}
Let $\bm{t}$ be a comb generated by $(\bm{\gamma},\bm{\delta})$ that forms an infinite clique in a tree-regular relation $R \subseteq \T_\Sigma \times \T_\Sigma$ given as an ATA $\A$. 
There exist a coarsening $(\bm{\alpha},\bm{\beta})$ of $(\bm{\gamma},\bm{\delta})$ that generates a comb $\bm{s}$, 
runs $\rho(s_i,s_j)$ of $\A$ on $s_i \otimes s_j$, and context forests $\kappa_i, \lambda_i, \mu_{i,j}, \nu_{i,j}$ such that
\[\rho(s_i,s_j) = \kappa_1 \dots \kappa_{i-1} \lambda_i \mu_{i,i+1} \dots \mu_{i,j-1} \nu_{i,j}\]
is a decomposition compatible with
the $(\bm{\alpha},\bm{\beta})$-decomposition of $s_i \otimes s_j$
for all $i < j$.
\end{lemma}

\begin{proof}
The proof is similar to the proof of \cref{lem:comb-of-combs} in the word case.
We emphasize that the pigeonhole principle and Ramsey's theorem can be applied as in the word case
since the unique prefixes of the runs of the ATA that are runs on a given context have bounded size.
Further note that since a run on a forest is not defined, instead of considering only the suffixes $\tau(t_i,t_j)$ and $\sigma(t_i,t_j)$ in the inductive step,
we have to consider the whole run $\rho(t_i,t_j)$ and extend the common prefix that is already fixed.
\end{proof}

An \emph{alternating Büchi tree automaton} (ABTA) over the ranked alphabet $\Sigma$ is a tuple $\A = (Q,\Sigma,\delta,q_0,F)$ where $Q$, $\Sigma$, $\delta$, and $q_0$ are as in the definition of an ATA and $F \subseteq Q$ is a set of final states.
A \emph{run} of $\A$ on a finite or infinite tree $t \in \T_\Sigma^\infty$ is a tree $\rho \in \U_{\N^* \times Q}^\infty$ such that 
$\rho(\varepsilon) = (\varepsilon, q_0)$ and 
for all $u \in \dom(\rho)$ with $r$ children and $\rho(u) = (w,q)$ there is a satisfying assignment $S = \{(q_1,c_1),\dots,(q_r,c_r)\} \models \delta(q,t(w))$
of pairwise distinct $(q_i,c_i)$ such that $\rho(ui) = (wc_i, q_i)$ for all $i \in [1,r]$.
A run $\rho$ is \emph{accepting} if every infinite path of $\rho$ contains infinitely many nodes with labels in $\N^* \times F$.

\begin{restatable}{proposition}{treebuchi}\label{lem:tree-buchi}
Given an ATA $\A$ for a tree-regular relation $R \subseteq (\T_\Sigma)^2$,
one can construct in polynomial time an ABTA $\B$ over the ranked alphabet $\Omega = \Sigma_\bot^2 \cup \{\#\}$ such that
we have:
\begin{itemize}
\item If $\bm{t}$ is an infinite clique in $R$,
then $\B$ accepts an encoding of a comb $\bm{s}$ which is a subsequence of $\bm{t}$.
\item If $\B$ accepts $t \in \T_\Omega^\infty$, then $t$ is an encoding of a comb $\bm{t}$ that is an infinite clique in $R$.
\end{itemize}
\end{restatable}

\begin{proof}

Let $\A = (Q,\Sigma_\bot^2,\delta,q_0)$ be the ATA that recognizes $R$.
We construct an ABTA $\B$ over $\Omega$ which accepts precisely
all comb encodings $\enc(\bm{\alpha},\bm{\beta})$ with the properties from \cref{lem:tree-run-comb}.
The state set of $\B$ is $Q \times \{1,2,3,4\}$, representing four different modes. 
In the first mode it simulates $\kappa_j$ on $\beta_j \otimes \beta_j$, 
in the second mode it simulates $\lambda_j$ on $\alpha_j \otimes \beta_j$, 
in the third mode it simulates $\mu_{i,j}$ on $\varepsilon \otimes \beta_j$, 
and in the fourth mode it simulates $\nu_{i,j}$ on $\varepsilon \otimes \alpha_j$.
\Cref{fig:buchi-simulation} illustrates the simulation.

For all $q \in Q$ and $\svector{a}{b} \in \Sigma_\bot^2$ we set
\begin{align*}
\delta_\B((q,1),\svector{a}{b}) &:= \delta_1(q,\svector{b}{b}) &
\delta_\B((q,2),\svector{a}{b}) &:= \delta_2(q,\svector{a}{b}) \\
\delta_\B((q,3),\svector{a}{b}) &:= \delta_3(q,\svector{\bot}{b}) &
\delta_\B((q,4),\svector{a}{b}) &:= \delta_4(q,\svector{\bot}{a})
\end{align*} 
where $\delta_i(q,\sigma)$ is the formula $\delta(q,\sigma)$ where each variable $(p,c)$ is replaced by $((p,i),c)$.
At the holes of $\alpha_j \otimes \beta_j$ labeled with $\#$, i.e., the points where $\alpha_{j+1} \otimes \beta_{j+1}$ starts, the simulations in modes 1, 2, and 3 split up.
For all $q \in Q$ we define
\begin{align*}
\delta_\B((q,1),\#) &:= ((q,1),1) \wedge ((q,2),1) \\
\delta_\B((q,2),\#) := \delta_\B((q,3),\#) &:= ((q,3),1) \wedge ((q,4),1)
\end{align*}
Finally, we add a new initial state $q_0^\B$ which spawns simulations of $\A$ in mode 1 and 2,
i.e., for all $q \in Q$ and $\svector{a}{b} \in \Sigma_\bot^2$ we define
\[\delta_\B(q_0^\B,\svector{a}{b}) := \delta_1(q_0,\svector{b}{b}) \wedge \delta_2(q_0,\svector{a}{b}).\]
Finally, we intersect $L(\B)$ with the tree-regular language $\Enc_\Sigma$,
which concludes the proof.
\end{proof}

If $R$ is given by a D$\uparrow$TA we can even compute in polynomial time
a {\em nondeterministic} Büchi tree automaton (NBTA) for the representation of infinite cliques.
The proof idea is that the runs $\mu_{i,j}$ and $\nu_{i,j}$ in \cref{lem:tree-run-comb} only depend on $j$.

\begin{restatable}{proposition}{dettreebuchi}\label{lem:det-tree-buchi}
Given a D$\uparrow$TA $\A$ for a tree-regular relation $R \subseteq (\T_\Sigma)^2$,
one can construct in polynomial time an NBTA $\B$ over the ranked alphabet $\Omega = \Sigma_\bot^2 \cup \{\#\}$ such that
we have:
\begin{itemize}
\item If $\bm{t}$ is an infinite clique in $R$,
then $\B$ accepts an encoding of a comb $\bm{s}$ which is a subsequence of $\bm{t}$.
\item If $\B$ accepts $t \in \T_\Omega^\infty$, then $t$ is an encoding of a comb $\bm{t}$ that is an infinite clique in $R$.
\end{itemize}
\end{restatable}

We are ready to prove \cref{thm:tree-ramsey}.
If $\A$ is an ATA, we use the ABTA $\B$ from \Cref{lem:tree-buchi}
and transform it into an NBTA in exponential time \cite[Theorem 1.2]{muller1995simulating}.
If $\A$ is a D$\uparrow$TA $\A$ we use the NBTA $\B$ from \Cref{lem:det-tree-buchi}.
We then construct an NTA $\C$ over $\Sigma$ which accepts $t_1 \in \T_\Sigma$
if and only if the encoding of some comb $(t_i)_{i \ge 1}$ is accepted by $\B$.
The rest of the proof is analogous to the proof of \Cref{thm:word-ramsey},
see \cref{sec:app-buchi} for details.

%% file: transitive.tex
In this section we show that if we assume that $R$ is transitive, then the Ramsey quantifier can be evaluated in polynomial time (\cref{thm:transitive-ramsey}).

\begin{restatable}{proposition}{transitivebuchi}\label{lem:transitive-buchi}
Given an NTA $\A$ for a transitive tree-regular relation $R \subseteq (\T_\Sigma)^2$,
one can construct in polynomial time an NBTA $\B$ over the ranked alphabet $\Omega$ such that:
\begin{itemize}
\item If $\bm{t}$ is an infinite clique in $R$,
then $\B$ accepts an encoding of a comb $\bm{s}$ which is a subsequence of $\bm{t}$.
\item If $\B$ accepts $t \in \T_\Omega^\infty$, then $t$ is an encoding of a comb $\bm{t}$ that is an infinite clique in $R$.
\end{itemize}
\end{restatable}
For the proof we view $\A$ as an ATA
and construct the ABTA $\B$ as in the proof of \Cref{lem:tree-buchi}
which accepts precisely all comb encodings $\enc(\bm{\alpha},\bm{\beta})$ with the properties from \cref{lem:tree-run-comb}.
Since $R$ is transitive, ensuring that $\bm{v}$ is a clique merely requires to check that $(v_i,v_{i+1})\in R$ for each~$i$.
Therefore, the set of runs we need to detect on an encoding $\enc(\bm{\alpha},\bm{\beta})$ are as in \cref{fig:buchi-simulation}, but without all
the runs $\mu_{i,j}$. In the resulting comb of combs, all rows have finite length.
In terms of the constructed ABTA, this means we can omit all states in mode 3.
Then any run contains for each node $v$ of the input tree at most three run nodes referring to $v$.
Thus, we can apply a standard powerset construction to convert $\C$ into an equivalent NBTA of polynomial size
by restricting to subsets of states of size at most three.
We refer to \cref{sec:app-buchi} for more details.
Using the polynomially-sized NBTA $\B$
we can prove \cref{thm:transitive-ramsey} analogously to \Cref{thm:tree-ramsey}.

%% file: cotransitive.tex
Recall that a binary relation $R \subseteq A \times A$ is {\em co-transitive}
if its complement $\bar R = (A \times A) \setminus R$ is transitive.
Next we show \cref{thm:cotransitive}.

A context $\beta$ is called {\em monadic} if it has exactly one hole.
We will show that, if a co-transitive relation has an infinite clique, then there exists one which is a comb
generated by a {\em monadic generator} $(\bm{\alpha}, \bm{\beta})$
in which all forests $\beta_i$ are monadic contexts.
This also implies that all $\alpha_i$ are trees.

\begin{lemma}
	\label{lem:cotrans-step}
	If a co-transitive tree-regular relation $R \subseteq \T_\Sigma \times \T_\Sigma$ has an infinite clique
	over a tree-regular language $L \subseteq \T_\Sigma$,
	then there exists an infinite clique $\bm{t}$ of $R$ over $L$
	and a nontrivial monadic context $\beta$ with $t_1 \triangleleft \beta$
	and $\beta \pref t_i$ for all $i \ge 2$.
\end{lemma}

\begin{proof}
Let $\A = (Q,\Sigma_\bot^2,\Delta,F)$ and $\B = (P,\Sigma,\Lambda,G)$
be D$\uparrow$TAs for $R$ and $L$, respectively.
Suppose that $\bm{t}$ is an infinite clique in $R$ over $L$.
For $j \ge 2$ let $c_j$ be the unique context with
$\nodes(c_j) = \nodes(t_1) \cap \nodes(t_j)$ and $c_j \pref t_j$.
Notice that $c_j$ is nontrivial since $t_1$ and $t_j$ contain the root.
Furthermore we have $t_1 \triangleleft c_j$ since any hole $u \in \holes(c_j)$
is contained in $\nodes(t_j) \setminus \nodes(c_j)$.
Since there are only finitely many such choices for $c_j$, by reducing $\bm{t}$ to a subsequence which starts with $t_1$
we can assume that $c_j = c$ for all $j \ge 2$ for some context $c$.
Suppose that $v^1, \dots, v^n$ are the holes of $c$ in lexicographical order,
and $t_j = c[t_j^1, \dots, t_j^n]$ for some trees $t_j^k$.
Again, by reducing $\bm{t}$ to a subsequence starting with $t_1$,
we can further assume that $(t_j^k)_{j \ge 2}$ is an infinite sequence of pairwise distinct trees for each $k \in [1,n]$.
Indeed, if $(t_j^k)_{j \ge 2}$ contains only finitely many distinct trees for some $k$,
then some tree $t$ must occur infinitely often in the sequence $(t_j^k)_{j \ge 2}$,
say $t = t_{\ell_2}^k = t_{\ell_3}^k = \cdots$
for some $1 = \ell_1 < \ell_2 < \cdots$.
We then extend $c$ by plugging $t$ into the hole $v^k$ and we replace $\bm{t}$ by $(t_{\ell_j})_{j \ge 1}$.
Clearly, duplicates in a sequence $(t_j^k)_{j \ge 2}$ that contains infinitely many distinct elements can also be removed by restricting to a subsequence.

For all $i < j$ we have
\[
	t_i \otimes t_j =
	\begin{cases}
	(t_1 \otimes c)[\varepsilon \otimes t_j^1, \dots, \varepsilon \otimes t_j^n], & 1 = i < j, \\
	(c \otimes c)[t_i^1 \otimes t_j^1, \dots, t_i^n \otimes t_j^n], & 1 < i < j,
	\end{cases}
\]
where $t_1 \otimes c$ and $c \otimes c$ are naturally viewed as contexts with $n$ holes.
For $j \ge 2$ consider the accepting run $\rho_j$ of $\A$ on $t_1 \otimes t_j$
and the accepting run $\pi_j$ of $\B$ on $t_j$,
and color each index $j$ by the tuple $(\rho_j(v^1), \dots, \rho_j(v^n), \pi_j(v^1), \dots, \pi_j(v^n))$.
By the pigeonhole principle we can pick numbers $1 = \ell_1 < \ell_2 < \dots$
such that $\{\ell_2, \ell_3, \dots \}$ is monochromatic.
We then replace $\bm{t}$ by $(t_{\ell_i})_{i \ge 1}$.
Hence, the accepting runs of $\A$ on $t_1 \otimes t_j$ ($j \ge 2$) visit the same states $r^1, \dots, r^n$ in the nodes $v^1, \dots, v^n$.
Similarly, the accepting runs of $\B$ on the trees $t_j$ visit the same states $p^1, \dots, p^n$ in the nodes $v^1, \dots, v^n$.
Therefore
\begin{align*}
	(t_1, c[t_{j_1}^1, \dots, t_{j_n}^n]) &\in R & &\text{for any } j_1, \dots, j_n \ge 2~~\text{and} \\
	c[t_{j_1}^1, \dots, t_{j_n}^n] &\in L  & &\text{for any } j_1, \dots, j_n \ge 2.
\end{align*}
For $1 < i < j$ consider an accepting run of $\A$ on $t_i \otimes t_j$
and let $q_{i,j}^k$ be the state reached in node $v^k$.
By Ramsey's theorem we can assume that there exist states $q^1, \dots, q^n \in Q$
such that $q_{i,j}^k = q^k$ for all $1 < i < j$
(again, after replacing $\bm{t}$ by a subsequence starting with $t_1$).
Observe that $\A$ accepts the context $c \otimes c$
if it starts in nodes $v^1, \dots, v^k$ with the states $q^1, \dots, q^k$, respectively.

For every $0 \le k \le n$ define the tree 
\[s_k = c[t_3^1, \dots, t_3^k, t_2^{k+1}, \dots, t_2^n].\]
We have $(s_0,s_n) = (t_2,t_3) \in R$.
There must be an index $1 \le k \le n$ with $(s_{k-1},s_k) \in R$
since otherwise by transitivity of $\bar R$ we would have $(s_0,s_n) \notin R$.
Define the context
\[\beta = c[t_3^1, \dots, t_3^{k-1}, x, t_2^{k+1}, \dots, t_2^n].\]
Then we have that
$(\beta[t_2^k],\beta[t_3^k]) \in R$.
This is witnessed by an accepting run on their convolution,
which reaches state $q^k$ at node $v^k$.
This implies that
$(\beta[t_i^k],\beta[t_j^k]) \in R$
for all $i < j$, since the run of $\A$ on $t_i^k \otimes t_j^k$ also reaches $q^k$.
Hence, the context $\beta$
together with the trees $t_1$ and $\beta[t_i^k]$ for $i \ge 2$ satisfy the claim.
Moreover, $t_1 \triangleleft c$ and $c \pref \beta$ implies $t_1 \triangleleft \beta$.
\end{proof}

Repeated applications of \Cref{lem:cotrans-step} yields the desired infinite clique:

\begin{lemma}
	\label{lem:cotrans-cliques}
	If a co-transitive tree-regular relation $R$ has an infinite clique
	then there exists an infinite clique $\bm{t}$ of $R$
	generated by a monadic generator $(\bm{\alpha}, \bm{\beta})$.
\end{lemma}

\begin{proof}
	Let $n \in \N$ and
	suppose we have inductively constructed trees $\alpha_1, \dots, \alpha_n$,
	and nontrivial monadic contexts $\beta_1, \dots, \beta_n$,
	with $\alpha_i \triangleleft \beta_i$ for all $1 \le i \le n$,
	such that there exist trees $(t_i')_{i > n}$
	such that $(t_i)_{i \ge 1}$ is an infinite clique in $R$ where
	$t_i = \beta_1 \beta_2 \dots \beta_{i-1} \alpha_i$ if $i \le n$,
	and $t_i = \beta_1 \beta_2 \dots \beta_n t_i'$ if $i > n$.
	Let $\beta := \beta_1 \beta_2 \dots \beta_{n}$.
	Then $(t_i')_{i > n}$ is an infinite clique in the relation
	$R' = \{ (s',t') \mid (\beta[s'],\beta[t']) \in R \}$.
	It is easy to see that $R'$ is again tree-regular and also co-transitive
	since transitivity is preserved from $\bar R$ to $\bar R'$ via $\beta[t] \mapsto t$.
	Furthermore all trees $t_i'$ for $i > n$ belong to the tree-regular language
	$L = \bigcap_{i=1}^n \{ t \mid (t_i,\beta[t]) \in R \}$.
	We can apply \Cref{lem:cotrans-step} and obtain a tree $\alpha_{n+1}$,
	a nontrivial monadic context $\beta_{n+1}$
	with $\alpha_{n+1} \triangleleft \beta_{n+1}$, and trees $(t_i'')_{i > n+1}$ such that
	$\alpha_{n+1}$ together with $\beta_{n+1}[t_i'']$ for $i > n + 1$ form an infinite clique in $R'$.
	Furthermore all trees $\beta_{n+1}[t_i'']$ for $i > n+1$ belong to $L$.
	Hence $t_1, \dots, t_n$ together with $\beta[\alpha_{n+1}]$ and $\beta[\beta_{n+1}[t_i'']]$ for $i > n+1$
	form an infinite clique in $R$.
	By induction we then obtain the desired sequences $\bm{\alpha}, \bm{\beta}$.
\end{proof}

We can now prove \Cref{thm:cotransitive}.
Given an NTA $\A$ for a co-transitive relation $R$.
Using \Cref{lem:cotrans-cliques,lem:tree-run-comb},
we can prove a statement similar to \Cref{lem:easy-combs} for tree combs which are generated by a monadic generator.
In particular, all context forests $\kappa_i, \lambda_i, \mu_j$ have exactly one hole,
and hence $\kappa_i, \lambda_i, \mu_j, \nu_j$ are in fact contexts.
Now we can construct in polynomial time a Büchi tree automaton $\B$
which accepts all comb encodings
$\enc(\bm{\alpha},\bm{\beta})$ of a monadic generator $(\bm{\alpha},\bm{\beta})$
for which runs of the form $\kappa_j, \lambda_j, \mu_j, \nu_j$ as above exist.
To this end, $\B$ consists of four components in which the runs $\kappa_j, \lambda_j, \mu_j, \nu_j$ are simulated.
The detailed construction can be found in \cref{sec:app-buchi}.

The lower bound of \Cref{thm:cotransitive} follows by a logspace reduction from monadic decomposability for D$\uparrow$TAs which is $\P$-complete (see \Cref{sec:mondec}).

%% file: generalized-buchi.tex
\subsection{Recurrent reachability with generalized Büchi condition}

The proof of the following Proposition can be found in \cref{sec:clique-rec}.
\begin{proposition}\label{prop:clique-rec}
The infinite clique problem and recurrent reachability
are logspace equivalent over (tree-)regular relations.
Moreover, the logspace reduction from recurrent reachability to the infinite clique problem preserves 
transitivity of relations and determinism of automata.
\end{proposition}
Using \Cref{prop:clique-rec}
we obtain tight complexity bounds for recurrent reachability over (transitive) (tree-)regular relations.
We can even compute an automaton for the set $\rec(L)[R]$ of initial elements given automata for $R$ and $L$. 

\begin{corollary}\label{cor:tree-rec}
If $R$ is a binary (tree-)regular relation and $L$ is a (tree-)regular language given by NFAs (NTAs), 
then one can construct an NFA (NTA) for $\rec(L)[R]$ in logspace (exponential time).
The construction works in polynomial time if $R$ and $L$ are given by D$\uparrow$TAs
or if $R$ is transitive.
\end{corollary}

\begin{proof}
We can define $\rec(L)[R]$ by the formula
\begin{align*}
	\varphi(x) = \ram y,z \colon & R(x,y) \wedge L(y) \wedge R(y,z)  \\
	\vee \, \exists y \colon & R(x,y) \wedge L(y) \wedge R(y,y).
\end{align*}
Here, the first disjunct (beginning with $\ram$) captures infinite paths visiting infinitely many configurations,
whereas the second (beginning with $\exists$) captures infinite paths with only finitely many (distinct) configurations.

If $R$ and $L$ are given by NFAs, we can construct in logspace an NFA for $\rec(L)[R]$
using the closure properties of regular relations and \cref{thm:word-ramsey}.
Over trees, we use \cref{thm:tree-ramsey,thm:transitive-ramsey}
to construct an NTA for $\rec(L)[R]$ in exponential or polynomial time 
depending on whether $R$ is transitive and how $R,L$ are given.
\end{proof}

For recurrent reachability with generalized Büchi condition
we show that over words the complexity increases from $\NL$ to $\PSPACE$,
while over trees it stays in $\EXP$ (\cref{thm:general-buchi}).

For both the word and the tree case we reduce the generalized version to the classical version.
We first observe that $a_0 \in \rec(L_1,\dots,L_k)[R]$ if and only if there is a sequence $\bm{a}$ 
such that $(a_i,a_j) \in R$ for all $0 \le i < j$ and $a_i \in L_{((i-1) \text{ mod } k) + 1}$ for all $i \ge 1$.
We define a (tree-)regular relation $R' \subseteq A^k \times A^k$
that checks if a tuple $(a_1,\dots,a_{2k})$ forms a clique of size $2k+1$ in $R$ starting with $a_0$
such that $a_i \in L_i$ for all $i \in [1,k]$.
In the word case the NFA $\A'$ for $R'$ can be constructed in $\PSPACE$ using a product construction.
In the tree case we can avoid the exponential blow-up for the product automaton by using ATAs.
To make this work, we have to reduce the size of the alphabet for the ATA $\A'$.
This can be achieved by encoding a tuple $(\sigma_1,\dots,\sigma_k) \in \Sigma_\bot^k$ of symbols 
by a path $\sigma_1(\sigma_2(\dots \sigma_k(\#_m)\dots))$ where $\#_m$ is used as delimiter symbol of rank $m := \max\{\rk(\sigma_i) \mid 1 \le i \le k\}$.
Then the ATA $\A'$ can be constructed in polynomial time.
Now it holds that $a_0 \in \rec(L_1,\dots,L_k)[R]$ if and only if 
$\ram x,y \colon R'(x,y) \vee \exists x \colon R'(x,x)$
is valid.
By \cref{thm:word-ramsey} (resp. \cref{thm:tree-ramsey}) validity of the first disjunct of $\varphi$ 
can be checked in nondeterministic logspace (resp. exponential time) given $\A'$.
It is easy to see that validity of the second disjunct of $\varphi$ 
can also be checked in nondeterministic logspace (resp. exponential time) given $\A'$.
This yields a $\PSPACE$-algorithm in the word case and an $\EXP$-algorithm in the tree case.
Details are in \cref{sec:app-generalized}.

For the lower bounds we reduce from the intersection non-emptiness problem of (tree-)regular languages $L_1,\dots,L_k \subseteq A$, 
which is known to be $\PSPACE$-complete over words \cite{Kozen77} and $\EXP$-complete over trees \cite[Theorem~11]{tata}.
We define the (tree-)regular relation $R \subseteq A \times A$ such that $(a,b) \in R$ if and only if
$a=c$ or
$a=b$
where $c \in A$ is some fixed element.
Then $L_1 \cap \dots \cap L_k \ne \emptyset$ if and only if $c \in \rec(L_1,\dots,L_k)[R]$.

For $k=1$, the previous construction yields a reduction from nonemptiness for D$\uparrow$TAs,
which is $\P$-complete, to recurrent reachability over transitive tree-regular relations given by D$\uparrow$TAs,
proving the $\P$-hardness in \Cref{thm:tree-lowerbound} and \Cref{cor:rec-reach}.

The lower bound in the word case of \Cref{cor:rec-reach} follows by a logspace reduction from monadic decomposability for DFAs which is $\NL$-complete (see \Cref{sec:mondec}).

In \cite{Loding06} Löding shows that the reachability relation $\to^*$ for regular ground tree rewrite systems (RGTRS) is tree-regular
and an NTA for $\to^*$ can be constructed in polynomial time.
Hence, by \cref{thm:general-buchi} recurrent reachability with generalized Büchi condition is $\EXP$-complete for RGTRSs
where hardness for GTRSs can be shown by a similar reduction as above from intersection nonemptiness.

\begin{corollary}
Given an RGTRS and NTAs for tree-regular languages $L_1,\dots,L_k$,
one can construct in exponential time an NTA recognizing $\rec(L_1,\dots,L_k)[\to^*]$.
\end{corollary}

\subsection{Monadic Decomposability}\label{sec:mondec}

In the following we reduce monadic decomposability to the infinite clique problem over co-transitive relations,
proving \cref{cor:word-recognizable} and \cref{cor:recognizable}.
A $k$-ary relation $R$ over words or trees is monadically decomposable if and only if for all $1 \le j \le k$
the equivalence relations $\sim_j$ on $(\Sigma^*)^j$ (or $\T_\Sigma^j$) have finite index,
where two tuples $\bm{u} = (u_1, \dots, u_j)$, $\bm{v} = (v_1, \dots, v_j)$ are $\sim_j$-equivalent if and only if
\[
	\forall \bm{w} = (w_{j+1}, \dots, w_k) \colon [ (\bm{u},\bm{w}) \in R \iff (\bm{v},\bm{w}) \in R ],
\]
see for example \cite[Proof of Proposition~3.9]{CartonCG06}.
If the given automaton $\A$ for $R$ is a DFA, D$\uparrow$TA, or D$\downarrow$TA,
then one can compute automata $\A_{\not \sim_j}$ for the complements $\not \sim_j$ of $\sim_j$ in logspace,
using the fact that $\bm{u} \not \sim_j \bm{v}$ is equivalent to
\begin{align*}
	\exists \bm{w} \colon & \big( (\bm{u},\bm{w}) \in R \wedge (\bm{v},\bm{w}) \notin R \big) ~ \vee
	\big( (\bm{u},\bm{w}) \notin R \wedge (\bm{v},\bm{w}) \in R \big).
\end{align*}
If $\A$ is an NFA or NTA, then this is possible in polynomial space, by determinizing $\A$
and using closure properties of regular relations.
Then, apply \Cref{thm:word-ramsey} (\Cref{thm:cotransitive}) to $\A_{\not\sim_j}$
to check in $\NL$ (resp. $\P$) for an infinite clique in $\not \sim_j$.

We now prove the lower bounds by a reduction from the universality problem
for DFAs, NFAs, D$\uparrow$TAs, and NTAs,
and the emptiness problem for D$\downarrow$TAs.

\begin{lemma}\label{lem:hardness}
Given a binary regular relation $R \subseteq \Sigma^* \times \Sigma^*$ by an NFA (resp. DFA),
it is $\PSPACE$-hard (resp. $\NL$-hard) to decide whether $R$ is monadically decomposable.
Given a binary tree-regular relation $R \subseteq \T_\Sigma \times \T_\Sigma$ by an NTA (resp. D$\uparrow$TA),
it is $\EXP$-hard (resp. $\P$-hard) to decide whether $R$ is monadically decomposable.
\end{lemma}
\begin{proof}
We give a logspace reduction from the universality problem which is known to be $\PSPACE$-complete for NFAs,
$\NL$-complete for DFAs, $\P$-complete for D$\uparrow$TAs, and $\EXP$-complete for NTAs.
To ease notation, we only consider the word case and remark that the tree case is analogous.
Recall that the universality problem asks whether for a given regular language $L \subseteq \Sigma^*$ it holds that $L = \Sigma^*$.
Let $L \subseteq \Sigma^*$ be a regular language given by an NFA (resp. DFA) $\A$.
We define the regular relation
\[R_L := \{(u \otimes v, w) \mid u \in L \text{ or } v = w \in \Sigma^*\}.\]
It is easy to construct an NFA (resp. DFA) that recognizes $R_L$ in logarithmic space from $\A$. 
Note that for DFAs the disjunction can be realized with a product construction.
It remains to show that $R_L$ is monadically decomposable if and only if $L = \Sigma^*$.

If $L = \Sigma^*$, it holds that $R_L = \{(u \otimes v, w) \mid u,v,w \in \Sigma^*\}$ which is clearly monadically decomposable.

For the converse assume that there exists $u_0 \in \Sigma^* \setminus L$.
Then the intersection of $R_L$ with the monadically decomposable relation 
$\{(u_0 \otimes v, w) \mid v,w \in \Sigma^*\}$
is the relation
$\{(u_0 \otimes v, w) \mid v = w \in \Sigma^*\}$,
which is not monadically decomposable.
Since monadically decomposable relations are closed under intersection, it follows that $R_L$ is not monadically decomposable.
\end{proof}

\begin{lemma}\label{lem:det-hard}
Given a binary tree-regular relation $R \subseteq \T_\Sigma \times \T_\Sigma$ by a D$\downarrow$TA,
it is $\P$-hard to decide whether $R$ is monadically decomposable.
\end{lemma}
\begin{proof}
We give a logspace reduction from the emptiness problem for D$\downarrow$TAs which is known to be $\P$-complete \cite{veanes1997computational}. 
Let $\A = (Q,\Sigma,\Delta,q_0)$ be a D$\downarrow$TA. 
We construct a D$\downarrow$TA $\A' = (Q,\Gamma_\bot^2,\Delta',q_0)$ recognizing a binary tree-regular relation $R'$ over 
$\Gamma := \Sigma \cup \{\#\}$ where $\# \notin \Sigma$ is a symbol of rank 1 as follows. 
We define the transition relation $\Delta'$ such that
\begin{itemize}
\item $q \xrightarrow{(a,a)}_{\A'} (q_1,\dots,q_r)$ for all $q \xrightarrow{a}_{\A} (q_1,\dots,q_r)$,
\item $q_0 \xrightarrow{(\#,\#)}_{\A'} q_0$
\end{itemize}
Clearly, $\A'$ can be constructed in logspace from $\A$.

It is easy to see that $R' \subseteq \{(t,t) \mid t \in \T_{\Gamma}\}$. 
Moreover, it holds that $R'$ is finite if and only if $L(\A) = \emptyset$.  
Indeed, if there exists $t \in L(\A)$, then $(t_n,t_n) \in R'$ for all $n \geq 0$ 
where $t_n$ is the resulting tree when padding a chain of $\#$-symbols of length $n$ to the root of $t$. 
Since every finite relation is monadically decomposable and every infinite subrelation of $\{(t,t) \mid t \in \T_{\Gamma}\}$ is clearly not monadically decomposable, 
it holds that $R'$ is monadically decomposable if and only if $L(\A) = \emptyset$.
\end{proof}

%% file: unranked.tex
In this section we consider the unranked tree analogue of \cref{thm:tree-ramsey,thm:transitive-ramsey}.
Furthermore, we consider an application of the results to recurrent reachability in subtree and flat prefix rewriting systems.
{\em Unranked tree-regular} languages and relations are recognized by nondeterministic unranked tree automata (NUTAs),
see \Cref{sec:app-unranked}.
\begin{restatable}{theorem}{unrankedramsey}\label{thm:unranked-ramsey}
Given an unranked tree-regular $R \subseteq (\U_\Sigma)^{k+2}$
by an NUTA $\A$, an NUTA for the relation $\sem{\ram x,y \colon R(x,y,\bm{z})}{}$ can be constructed 
in polynomial time if $R$ is transitive and in exponential time otherwise.
Hence, the infinite clique problem over (transitive) unranked tree-regular relations
is in $\EXP$ ($\P$).
\end{restatable}
The proof can be found in \cref{sec:app-unranked-thm}.
It uses the \emph{first-child next-sibling encoding},
a standard regularity-preserving transformation from unranked trees to binary trees 
(see e.g.,~\cite{neven2002automata,gottlob2005complexity,libkin2005logics}).
As over ranked trees, \cref{thm:unranked-ramsey} implies:

\begin{corollary}\label{cor:unranked-rec}
For a binary unranked tree-regular relation $R$ and an unranked tree-regular language $L$ given by NUTAs, 
one can construct an NUTA recognizing $\rec(L)[R]$ in polynomial time if $R$ is transitive and in exponential time otherwise.
\end{corollary}

In \cite{loding2007transition} L{\"o}ding and Spelten introduce tree rewriting systems over unranked trees called subtree and flat prefix rewriting systems (SFPRS).
We refer to \cref{sec:app-sfprs} for the definition. 
In \cite{loding2007transition} it is shown that the reachability relation $\to^*$ for (regular) SFPRSs is an unranked tree-regular relation.
Moreover, it can be observed that the NUTA for $\to^*$ can be constructed in polynomial time.
Since $\to^*$ is transitive, we can apply \cref{cor:unranked-rec} to obtain that recurrent reachability for (regular) SFPRSs is decidable in polynomial time:

\begin{corollary}
For a (regular) SFPRS and an unranked tree-regular language $L$ given as NUTA, 
one can construct an NUTA recognizing $\rec(L)[\to^*]$ 
in polynomial time. 
\end{corollary}

Let $\mathrm{FO}(\mathrm{SFPRS})$ be the first-order theory over unranked trees with the reachability relation $\to^*$ and the one-step reachability relation $\to$ for (regular) SFPRSs. In \cite{loding2007transition} it is shown that the structure of $\mathrm{FO}(\mathrm{SFPRS})$ is unranked tree-automatic which means that $\mathrm{FO}(\mathrm{SFPRS})$ is decidable.

Let $\mathrm{FO}(\mathrm{SFPRS} + \mathrm{Rec})$ be the theory $\mathrm{FO}(\mathrm{SFPRS})$ enriched by the recurrent reachability operator:
For a formula $\varphi$ in $\mathrm{FO}(\mathrm{SFPRS} + \mathrm{Rec})$ with one free variable we define the recurrent reachability operator $\rec(\varphi)$ as formula with one free variable such that $\rec(\varphi)(t)$ is true if and only if $t \in \rec(L)[\to^*]$ for any $t \in \U_\Sigma$ where $L$ is the unranked tree-regular language defined by $\varphi$.

\begin{corollary}
The theory $\mathrm{FO}(\mathrm{SFPRS} + \mathrm{Rec})$ is decidable.
\end{corollary} 
For future work, we propose to investigate if  \cref{cor:unranked-rec}
could be applied to
other classes of tree rewriting systems over unranked trees, e.g., the class
$\textsl{Trs}_0$ of tree rewrite systems with positive guards 
\cite{HLO15}, which was
applied for the analysis of HTML5 applications.

%% file: conc.tex
\section{Conclusion and Future Works}
\label{sec:conc}

We have identified directed Ramsey quantifiers as a fundamental
notion that underlies the standard notion of Ramsey quantifiers, recurrent
reachability, and monadic decomposability. We have also shown that the notion of
\emph{comb of combs} can be used to obtain substantially simpler proofs in 
case of word-automatic relations, and can be generalized to tree-automatic
relations, allowing us to derive new results for Ramsey quantifiers, recurrent
reachability and monadic decomposability (with applications to generalized Büchi
conditions and unranked tree-automatic relations). There are many natural
research directions. In particular, we pinpoint that Ramsey quantifiers over
$\omega$-automatic relations, as well as recurrent reachability over transitive 
$\omega$-automatic relations, is still a major open problem \cite{K10}, although
monadic decomposability is known to be decidable \cite{LS19}. One possible 
approach is to consider the subclass of $\omega$-automatic relations that are
definable over the theory of mixed integer-real linear arithmetic $\langle
\mathbb{R}; \mathbb{Z}, 1,0,<,+\rangle$, for which the problem of Ramsey
quantifiers and recurrent reachability, to be the best our knowledge, is still
an open problem.

%% file: app-undirected.tex
\Cref{thm:word-ramsey} and \Cref{thm:tree-ramsey} also hold if we alternatively define the Ramsey quantifier
$\ram x, y \colon \varphi(x,y,\bm{z})$ using {\em infinite undirected cliques}, i.e., there exists an infinite set $X$
such that $\varphi(a,b,\bm{z})$ holds for all $a, b \in X$ with $a \neq b$,
since we can replace $\varphi(x,y,\bm{z})$ by $\varphi(x,y,\bm{z}) \wedge \varphi(y,x,\bm{z})$.

Furthermore, the \NL-lower bound in the word case (\cref{cor:word-recognizable})
and the \EXP-lower bound in the tree case (\Cref{thm:tree-lowerbound})
also hold for undirected cliques:

\begin{proposition}\label{prop:directed-undirected}
The infinite clique problems for directed and undirected cliques are logspace equivalent over (tree-)regular relations.
\end{proposition}

\begin{proof}
We first reduce the undirected version to the directed version. Let $R \subseteq A \times A$ be given by an NFA (resp. NTA) $\A$. Then we define the relation $R' := \{(a,b) \in A \times A \mid (a,b) \in R \wedge (b,a) \in R\}$. Clearly, $R'$ is a (tree-)regular relation and an NFA (resp. NTA) recognizing $R'$ can be constructed in logspace from $\A$. Moreover, we have that $R$ has an infinite undirected clique if and only if $R'$ has an infinite directed clique.

For the reverse reduction, let $R \subseteq A \times A$ be given by an NFA (resp. NTA) $\A$. We define the relation
\begin{multline*}
R' := \{((a,i),(b,j)) \in (A \times \N)^2 \mid a \ne b \\ \wedge ((a,b) \in R \wedge i < j \vee (b,a) \in R \wedge j < i)\}.
\end{multline*}
It is easy to see that $R'$ can be encoded as a (tree-)regular relation and an NFA (resp. NTA) recognizing this relation can be constructed in logspace from $\A$. It holds that $R$ has an infinite directed clique if and only if $R'$ has an infinite undirected clique. Indeed, if $R$ has an infinite directed clique $(a_i)_{i \ge 1}$, then we can number the elements and get an infinite undirected clique $((a_i,i))_{i \ge 1}$ in $R'$. Conversely, if $R'$ has an infinite undirected clique $((a_i,n_i))_{i \ge 1}$, then $(a_{i_j})_{j \ge 1}$ with $n_{i_j} < n_{i_{j'}}$ for $j < j'$ is an infinite directed clique in $R$.
\end{proof}

%% file: app-buchi.tex
\dettreebuchi*

\begin{proof}
Let $\A' = (Q,\Sigma_\bot^2,\Delta,q_{\mathrm{in}})$ be the NTA that is obtained by reverting the transitions of the D$\uparrow$TA $\A$.
Clearly, $\A'$ has the same runs as $\A$ on trees.
Let $\bm{t}$ be a comb that forms an infinite clique in $R$.
Since any NTA is a special ATA, we can apply \cref{lem:tree-run-comb} on $\A'$ and $\bm{t}$
to get a subcomb $\bm{s}$ of $\bm{t}$ generated by $(\bm{\alpha},\bm{\beta})$,
runs $\rho(s_i,s_j)$ of $\A'$ on $s_i \otimes s_j$, and context forests $\kappa_i, \lambda_i, \mu_{i,j}, \nu_{i,j}$ such that
\[\rho(s_i,s_j) = \kappa_1 \dots \kappa_{i-1} \lambda_i \mu_{i,i+1} \dots \mu_{i,j-1} \nu_{i,j}\]
is a decomposition compatible with
the $(\bm{\alpha},\bm{\beta})$-decomposition of $s_i \otimes s_j$
for all $i < j$.
Moreover, we have that $\mu_{i,j}$ and $\nu_{i,j}$ only depend on $j$
since there are unique runs of $\A$ on 
$(\varepsilon \otimes \beta_j)(\varepsilon \otimes \alpha_{j+1})$ and $\varepsilon \otimes \alpha_{j}$.
Thus, we can just write $\mu_j$ and $\nu_j$ for all $j > 1$.

We now construct an NBTA $\B$ over the alphabet $\Omega$ which accepts precisely all comb encodings $\enc(\bm{\alpha},\bm{\beta})$
of a generator $(\bm{\alpha},\bm{\beta})$ with the above properties.
Since the set of all comb encodings is regular, we can assume that the input tree is already a valid comb encoding.
A state in $\B$ consists of four components in which $\kappa_j, \lambda_j, \mu_j, \nu_j$ are simulated.
To handle the special case where only $\kappa_0$, $\lambda_0$ are simulated,
we add a state $\bot$ to $\A'$ with transitions
$\bot \xrightarrow{(a,b)} \bm{q}$ for all symbols $(a,b) \in \Sigma_\bot^2$ of rank $r$, and $\bm{q} \in Q_\bot^r$.
The NBTA $\B$ has the state set $Q_\bot^4$,
initial state $(q_{\mathrm{in}},q_{\mathrm{in}},\bot,\bot)$, and the transitions
\begin{itemize}
\item $(p,s,q,t) \xrightarrow{(a,b)} \bm{p} \otimes \bm{s} \otimes \bm{q} \otimes \bm{t}$
if $\A'$ contains the transitions
$p \xrightarrow{(b,b)} \bm{p}$,
$s \xrightarrow{(a,b)} \bm{s}$,
$q \xrightarrow{(\bot,b)} \bm{q}$,
$t \xrightarrow{(\bot,a)} \bm{t}$,
\item $(p,q,q,\bot) \xrightarrow{\#} (p,p,q,q)$ for all $p,q \in Q$.
\end{itemize}
Correctness follows from the previous observations.
\end{proof}

\treeramsey*

\begin{proof}
First observe that from the ABTA $\B$ in \cref{lem:tree-buchi} we can construct in exponential time an NTA $\C$ over $\Sigma$ which accepts $t_1 \in \T_\Sigma$ if and only if the encoding of some comb $(t_i)_{i \ge 1}$ is accepted by $\B$.
Indeed, we first transform $\B$ into an NBTA $\D = (Q_\D,\Omega,\Delta_D,q_0^\D,F_\D)$ which can be done in exponential time \cite[Theorem 1.2]{muller1995simulating}.
From $\D$ we construct $\C = (Q_\D,\Sigma,\Delta_\C,q_0^\D)$ such that for all $q \in Q_\D$, $a \in \Sigma$, and $p_i \in Q_\D$ for $1 \le i \le \rk(a)$ we let
$(q, a, (p_i)_{i \le \rk(a)}) \in \Delta_\C$
if and only if there exist $b \in \Sigma_\bot$ and $p_i \in Q_\D$ for $\rk(a) < i \le \rk\svector{a}{b}$ such that
\[(q, \svector{a}{b}, (p_i)_{i \le \rk\svector{a}{b}}) \in \Delta_\D\]
and $\D$ accepts some tree from state $p_i$ for all $i > \rk(a)$.
Note that $\C$ can be constructed in polynomial time given $\D$ since we need to perform a polynomial number of non-emptiness checks on $\D$, each of which takes quadratic time.
The NTA $\C$ satisfies that (i) for every infinite clique $\bm{t}$ of $R$ some element $t_i$ is accepted by $\C$ and 
(ii) if $t$ is accepted by $\C$, then $t$ belongs to an infinite clique of $R$.

Given an NTA $\A$ for $R \subseteq (\T_\Sigma)^{k+2}$.
We first construct an NTA $\A'$ over $\Sigma_\bot^{2k+2}$
accepting the binary relation
\[
	R' = \{ (s \otimes c_1 \otimes \dots \otimes c_k, t \otimes c_1 \otimes \dots \otimes c_k) \mid (s,t,\bm{c}) \in R \}.
\]
Let $\C'$ be the NTA described above that accepts at least one tree from each infinite $R'$-clique
and only accepts elements of infinite $R'$-cliques.
Projecting away the first component yields the desired NTA for $\{ \bm{c} \in (\T_\Sigma)^k \mid \ram x,y \colon R(x,y,\bm{c}) \}$.

If $R$ is given as D$\uparrow$TA, then one can construct an NBTA in polynomial time instead of an ABTA
using \cref{lem:det-tree-buchi}.
Then $\C'$ can be constructed in polynomial time.
\end{proof}

\transitivebuchi*

\begin{proof}
We view the NTA $\A$ for $R$ as an ATA
and construct the ABTA $\B$ as in the proof of \Cref{lem:tree-buchi}
which accepts precisely all comb encodings $\enc(\bm{\alpha},\bm{\beta})$ with the properties from \cref{lem:tree-run-comb}.
Then, we omit all states in mode 3.
More formally, let $\C = (Q_{\C},\Omega,\delta_{\C},q_0^{\C},Q_{\C})$ be the ABTA with state set
$Q_{\C} = \{ q_0^{\C} \} \cup (Q \times \{1,2,4\})$,
and the same transitions as $\B$ except for
$
\delta_{\C}((q,2),\#) := ((q,4),1)
$
for all $q \in Q$.
Clearly, all comb encodings with the properties from \cref{lem:tree-run-comb} are still accepted by $\C$.
Conversely, if the encoding $\enc(\bm{\alpha},\bm{\beta})$ of a comb $\bm{t}$ is accepted by $\C$ then
for all $i \ge 1$, $\C$ simulates a run of $\A$ on $t_i \otimes t_{i+1}$ as argued in \Cref{lem:tree-buchi}.
By transitivity we obtain $(t_i,t_j) \in R$ for all $i < j$.

It remains to convert $\C$ into an NBTA of polynomial size.
Observe that $\C$ only universally branches in the root and at $\#$-nodes into a state of mode 1 and 2.
Furthermore, states of mode 2 transition to mode 4 when reading $\#$.
Hence, any run contains for each node $v$ of the input tree at most three run nodes referring to $v$.
Thus, we apply a standard powerset construction to convert $\C$ into an equivalent NBTA $\D$,
where we restrict to subsets of $Q_{\C}$ of size at most three.
We make all states in $\D$ final, since $\C$ accepts any tree with a run.
Finally, we take the product construction of $\D$ with an NBTA for $\Enc_\Sigma$, to obtain the desired NBTA in polynomial time.
\end{proof}

\cotransitive*

\begin{proof}
Given an NTA $\A = (Q,\Sigma_\bot^2,\Delta,q_{\mathrm{in}})$ for a co-transitive relation $R \subseteq \T_\Sigma \times \T_\Sigma$.
We can prove a statement similar to \Cref{lem:easy-combs} for tree combs which are generated by a monadic generator.

Suppose that $R$ has an infinite clique.
By \Cref{lem:cotrans-cliques} and \Cref{lem:tree-run-comb}
there exist an infinite clique $\bm{t}$ in $R$ with a monadic generator $(\bm{\alpha},\bm{\beta})$,
and context forests $\kappa_i, \lambda_i, \mu_{i,j}, \nu_{i,j}$ for $i < j$ such that
$\rho_{i,j} = \kappa_1 \dots \kappa_{i-1} \lambda_i \mu_{i,i+1} \dots \mu_{i,j-1} \nu_{i,j}$
is a decomposition compatible with the $(\bm{\alpha},\bm{\beta})$-decomposition of $t_i \otimes t_j$.
In particular, all context forests $\kappa_i, \lambda_i, \mu_{i,j}$ have exactly one hole,
and hence $\kappa_i, \lambda_i, \mu_{i,j}, \nu_{i,j}$ are in fact contexts.

Moreover, we can ensure that $\mu_{i,j} = \mu_{i',j}$ and $\nu_{i,j} = \nu_{i',j}$ for all $i < i' < j$
and can therefore just write $\mu_j$ and $\nu_j$ for all $j > 1$, respectively.
For the proof we can reason similarly as in \Cref{lem:easy-combs}
by applying Ramsey's theorem to ensure that all contexts $\mu_{i,j}$ carry the same state in the root.
This allows us to verify the runs using a polynomially sized NBTA on the comb encoding.

Finally, we can construct in polynomial time a Büchi tree automaton $\B$
over the alphabet $\Omega = \Sigma_\bot^2 \cup \{\#\}$
which accepts all comb encodings
$\enc(\bm{\alpha},\bm{\beta})$ of a monadic generator $(\bm{\alpha},\bm{\beta})$
for which runs of the form $\kappa_j, \lambda_j, \mu_j, \nu_j$ as above exist.
Since the set of all monadic comb encodings is regular, we can assume that the input tree is already
a valid monadic comb encoding.
A state in $\B$ consists of four components in which the runs $\kappa_j, \lambda_j, \mu_j, \nu_j$ are simulated.
To handle the special case where only $\kappa_0$, $\lambda_0$ are simulated,
we add a state $\bot$ to $\A$ with transitions
$\bot \xrightarrow{(a,b)} \bm{q}$ for all symbols $(a,b) \in \Sigma_\bot^2$ of rank $r$, and $\bm{q} \in Q_\bot^r$.
The Büchi tree automaton $\B$ has the state set $Q_\bot^4$,
initial state $(q_{\mathrm{in}},q_{\mathrm{in}},\bot,\bot)$, and the transitions
\begin{itemize}
\item $(p,s,q,t) \xrightarrow{(a,b)} \bm{p} \otimes \bm{s} \otimes \bm{q} \otimes \bm{t}$
if $\A$ contains the transitions
$p \xrightarrow{(b,b)} \bm{p}$,
$s \xrightarrow{(a,b)} \bm{s}$,
$q \xrightarrow{(\bot,b)} \bm{q}$,
$t \xrightarrow{(\bot,a)} \bm{t}$,
\item $(p,q,q,\bot) \xrightarrow{\#} (p,p,q,q)$ for all $p,q \in Q$.
\end{itemize}
Correctness follows from the previous observations.
\end{proof}

%% file: app-equivalence.tex
\begin{lemma}
The infinite clique problem is logspace reducible to recurrent reachability over (tree-)regular relations.
\end{lemma}
\begin{proof}
The word case is easy. Let $R \subseteq \Sigma^* \times \Sigma^*$ be given by an NFA $\A$. We have that $R$ has an infinite clique if and only if there exists a sequence $(w_i)_{i \ge 1}$ of words such that $(w_i,w_j) \in R$ and $|w_i| < |w_j|$ for all $1 \le i < j$. We define the relation $R' \subseteq \Sigma^* \times \Sigma^*$ such that
\begin{itemize}
\item $(\varepsilon,w) \in R'$ for all $w \in \Sigma^+$ and
\item $(v,w) \in R'$ iff $(v,w) \in R$ and $|v| < |w|$ for all $v,w \in \Sigma^+$.
\end{itemize} 
Clearly, the relation $R'$ is regular and an NFA that accepts $R'$ is implicitly logspace computable.
Since a path in $R'$ cannot visit a word more than once,
it holds that $R$ has an infinite clique if and only if $\varepsilon \in \rec(\Sigma^*)[R']$.

In the tree case we use a similar idea as in the word case for one path of the trees. Let $R \subseteq T_\Sigma \times T_\Sigma$ be given by an NTA $\A$. We have that $R$ has an infinite clique if and only if there exists a sequence $(t_i)_{i \ge 1}$ of trees such that $(t_i,t_j) \in R$ for all $1 \le i < j$ and the domain of one path strictly grows indefinitely in the sequence.
Such a sequence exists since the union of the domains of all $t_i$ is an infinite ranked tree domain which by K{\"o}nig's lemma has an infinite path.
Let $t_0 \colon \{\varepsilon\} \to \Sigma$ with $t_0(\varepsilon) := a$ for some $a \in \Sigma$.
We define the relation $R' \subseteq T_{\Sigma'} \times T_{\Sigma'}$ with $\Sigma' = \Sigma \cup \{a' \mid a \in \Sigma\}$ such that for all $s,t \in T_{\Sigma'} \setminus \{t_0\}$ we have
\begin{itemize}
\item $(t_0, t) \in R'$ and
\item $(s,t) \in R'$ iff the non-primed versions of $s$ and $t$ are in relation in $R$ and there exists exactly one path $p$ from the root to a leaf in $s$ labeled with primed symbols, exactly one path $q$ from the root to a leaf in $t$ labeled with primed symbols, and the domain of $p$ is a strict subpath of the domain of $q$.
\end{itemize}
The relation $R'$ is tree-regular and an NTA for $R'$ is implicitly logspace computable. Indeed, the NTA for $R'$ nondeterministically guesses the path of primed labels in the convolution and verifies that all other paths have non-primed labels. Moreover, it can check if the path in $s$ is padded and therefore a strict subpath of the path in $t$. Finally, it can simulate the automaton $\A$ for $R$ on the non-primed versions of $s$ and $t$ by just treating a symbol $a'$ as $a$.

We claim that $R$ has an infinite clique if and only if $t_0 \in \rec(T_{\Sigma'})[R']$. We already argued the ``only if'' direction. For the ``if'' direction note that the path property prevents the witnessing sequence from visiting the non-primed version of a tree more than once. 
\end{proof}

\begin{lemma}
Recurrent reachability is logspace reducible to the infinite clique problem over (tree-)regular relations.
\end{lemma}
\begin{proof}
We use the same idea for both the word and tree case. Let $R \subseteq A \times A$ be a (tree-)regular relation given by an NFA (resp. NTA) $\A$ and $L \subseteq A$ be a (tree-)regular language given by an NFA (resp. NTA) $\B$. Furthermore, let $a_0 \in A$ be the initial word (resp. tree). We define the relation $R' \subseteq (A \times \N) \times (A \times \N)$ such that $((a,m),(b,n)) \in R'$ if and only if
\begin{itemize}
\item $(a_0,a) \in R$,
\item $(a,b) \in R$, and
\item $a \in L$.
\end{itemize}
Intuitively, we create infinitely many copies of every word (resp. tree) by taking the direct product with the integers. This allows the witnessing sequence of the infinite clique to visit a word (resp. tree) several times. Furthermore, in $R'$ we only consider the words (resp. trees) that are in relation with $a_0$ to ensure that $a_0$ fulfills the conditions of the initial word (resp. tree). With the third condition we ensure that every word (resp. tree) of the infinite clique is contained in $L$. Thus, $a_0 \in \rec(L)[R]$ if and only if $R'$ has an infinite clique.

Note that $R'$ is (tree-)regular and an NFA (resp. NTA) for it is implicitly logspace computable. To this end, we represent the integers in unary as words (resp. paths) and take the convolution with the corresponding word (resp. tree). The first condition can be checked by hardwiring $a_0$ into the automaton. The second and third conditions can be ensured by simulating $\A$ and $\B$, respectively.

Note that if $R$ is transitive, then so is $R'$.
Moreover, if $\A$ and $\B$ are deterministic, then so is the automaton for $R'$.
\end{proof}

%% file: app-generalized.tex
\begin{proposition}\label{prop:general-buchi-words}
Recurrent reachability with generalized Büchi condition is decidable in polynomial space over words.
\end{proposition}
\begin{proof}
We give a $\PSPACE$-reduction from the generalized version to the classical version.
Let the relation $R \subseteq \Sigma^* \times \Sigma^*$ be given by an NFA $\A$, the languages $L_1,\dots,L_k \subseteq \Sigma^*$ be given by NFAs $\A_1,\dots,\A_k$, and $s_0 \in \Sigma^*$ be the initial word.
First observe that $s_0 \in \rec(L_1,\dots,L_k)[R]$ if and only if there is a sequence of words $(s_i)_{i \ge 1}$ such that $(s_i,s_j) \in R$ for all $0 \le i < j$ and $s_i \in L_{((i-1) \text{ mod } k) + 1}$ for all $i \ge 1$.
We define the relation
\[R_i := \{(w_1,\dots,w_{2k}) \in (\Sigma^*)^{2k} \mid w_i \in L_i\}\]
for all $i \in [1,k]$.
Moreover, for all $1 \le i < j \le 2k$ let
\[R_{i,j} := \{(w_1,\dots,w_{2k}) \in (\Sigma^*)^{2k} \mid (w_i,w_j) \in R\}.\]
Finally, we define the relation
\[R_{s_0,i} := \{(w_1,\dots,w_{2k}) \in (\Sigma^*)^{2k} \mid (s_0,w_i) \in R\}\]
for all $i \in [1,k]$.
Then $s_0 \in \rec(L_1,\dots,L_k)[R]$ if and only if 
\[\varphi := \ram x,y \colon R'(x,y) \vee \exists x \colon R'(x,x)\] 
is valid where
\[R' := \bigcap_{i=1}^{k} R_i \cap \bigcap_{1 \le i < j \le 2k} R_{i,j} \cap \bigcap_{i=1}^k R_{s_0,i} \subseteq (\Sigma^*)^k \times (\Sigma^*)^k.\]
Note that the product automaton $\A'$ that recognizes $R'$ can be constructed in polynomial space.
By \cref{thm:word-ramsey} validity of the first disjunct of $\varphi$ can be checked in nondeterministic logspace given $\A'$.
It is easy to see that validity of the second disjunct of $\varphi$ can also be checked in nondeterministic logspace given $\A'$.
This yields a $\PSPACE$-algorithm in total.
\end{proof}

\begin{proposition}\label{prop:general-buchi-trees}
Recurrent reachability with generalized Büchi condition is decidable in exponential time over trees.
\end{proposition} 
\begin{proof}
We proceed similarly to the word case but we use ATAs to avoid the exponential blow-up for the product automaton.
Let the relation $R \subseteq \T_\Sigma \times \T_\Sigma$ be given by an ATA $\A = (Q,\Sigma_\bot^2,\delta,q_0)$, the languages $L_1,\dots,L_k \subseteq \T_\Sigma$ be given by ATAs $\A_i = (Q_i,\Sigma,\delta_i,q_0^i)$ for all $i \in [1,k]$, and $s_0 \in \T_\Sigma$ be the initial tree.
Note that we may assume that the relation and the languages are given by alternating automata since an NTA can be easily converted into an ATA in polynomial time.
Let $r := \max\{\rk(a) \mid a \in \Sigma\}$ and $\Omega := \Sigma_\bot \cup \{\#_i \mid 0 \le i \le r\}$ be a new ranked alphabet with $\rk(a) := 1$ for all $a \in \Sigma_\bot$ and $\rk(\#_i) := i$ for all $i \in [1,r]$.
For trees $t_1,\dots,t_n \in \T_\Sigma$ we define $p(t_1,\dots,t_n) \in \T_\Omega$ to be the tree $t_1 \otimes \dots \otimes t_n$ where each node labeled with $(a_1,\dots,a_n) \in \Sigma_\bot^n$ is replaced by a path $a_1(a_2(\dots a_n(\#_m)\dots))$ where $m := \max\{\rk(a_i) \mid 1 \le i \le n\}$.

Let 
\[R_p := \{(s,t) \in \T_\Omega^2 \mid s = p(t_1,\dots,t_k) \wedge t = p(t_{k+1},\dots,t_{2k})\}\] 
be the binary relation that checks if the trees are images under $p$.
Note that an ATA for $R_p$ can easily be constructed.
We define ATAs recognizing relations $R_i$ for all $i \in [1,k]$ and $R_{i,j}$ for all $1 \le i < j \le 2k$ with a similar meaning as in the word case. 
We start with the construction of the ATA 
\[\B_i = (Q_i^\B,\Omega^2,\delta_i^\B,(q_0,0))\]
for $R_i$. 
Intuitively, $\B_i$ checks if in $p(t_1,\dots,t_k) \otimes p(t_{k+1},\dots,t_{2k})$ we have that $t_i \in L_i$ for all $i \in [1,k]$.
The set of states of $\B_i$ is defined as
\[Q_i^\B := Q_i \times \{0,\dots,i-1\} \cup Q_i \times \Sigma \times \{i,\dots,k\}.\]
For all $q \in Q_i$, $j \in [0,k]$, $a,b \in \Sigma_\bot$, and $c \in \Sigma$ we let
\begin{align*}
\delta_i^\B((q,j),\svector{a}{b}) &:= ((q,j+1),1), \text{ if } j \le i-2 \\
\delta_i^\B((q,i-1),\svector{c}{b}) &:= ((q,c,i),1) \\
\delta_i^\B((q,c,j),\svector{a}{b}) &:= ((q,c,j+1),1), \text{ if } i \le j < k
\end{align*}
and for all $r_1,r_2 \in [0,r]$ with $\rk(c) \le \max\{r_1,r_2\}$ let
\[\delta_i^\B((q,c,k),\svector{\#_{r_1}}{\#_{r_2}}) := \delta'_i(q,c)\]
where $\delta'_i(q,c)$ is the formula $\delta_i(q,c)$ in which each variable $(p,\ell)$ is replaced by $((p,0),\ell)$.

We now construct the ATA 
\[\A_{k_1,k_2} = (Q_{k_1,k_2},\Omega^2,\delta_{k_1,k_2},(q_0,0))\]
for $R_{k_1,k_2}$.
Intuitively, $\A_{k_1,k_2}$ checks if $(t_{k_1},t_{k_2}) \in R$ holds in $p(t_1,\dots,t_k) \otimes p(t_{k+1},\dots,t_{2k})$.
We only show the construction for the case $1 \le k_1 \le k < k_2 \le 2k$ and $k_1 < k_2-k$ and note that the other cases work analogously.
The set of states of $\A_{k_1,k_2}$ is defined as
\begin{align*}
Q_{k_1,k_2} := \ & Q \times \{0,\dots,k_1-1\} \ \cup \\
& Q \times \Sigma_\bot \times \{k_1,\dots,k_2-k-1\} \ \cup \\
& Q \times \Sigma_\bot^2 \times \{k_2-k,\dots,k\}.
\end{align*}
We now define the transition function.
For all $q \in Q$, $j \in [0,k]$, and $a,b,c,d \in \Sigma_\bot$ we let
\[\delta_{k_1,k_2}((q,j),\svector{a}{b}) := ((q,j+1),1)\]
if $j \le k_1-2$,
\begin{align*}
\delta_{k_1,k_2}((q,k_1-1),\svector{a}{b}) &:= ((q,a,k_1),1) \\
\delta_{k_1,k_2}((q,c,j),\svector{a}{b}) &:= ((q,c,j+1),1)
\end{align*}
if $k_1 \le j \le k_2-k-2$,
\begin{align*}
\delta_{k_1,k_2}((q,c,k_2-k-1),\svector{a}{b}) &:= ((q,\svector{c}{b},k_2-k),1) \\
\delta_{k_1,k_2}((q,\svector{c}{d},j),\svector{a}{b}) &:= ((q,\svector{c}{d},j+1),1)
\end{align*}
if $k_2-k \le j \le k-1$, and for all $r_1,r_2 \in [0,r]$ with $\rk(c), \rk(d) \le \max\{r_1,r_2\}$ let
\[\delta_{k_1,k_2}((q,\svector{c}{d},k),\svector{\#_{r_1}}{\#_{r_2}}) := \delta'(q,\svector{c}{d})\]
where $\delta'(q,\svector{c}{d})$ is the formula $\delta(q,\svector{c}{d})$ in which each variable $(p,\ell)$ is replaced by $((p,0),\ell)$.

The ATA for the relation $R_{s_0,i}$ with $i \in [1,k]$ that checks if in $p(t_1,\dots,t_k) \otimes p(t_{k+1},\dots,t_{2k})$ we have that $(s_0,t_i) \in R$ can be constructed similarly to $\A_{i,j}$.
Note that all the constructions above can be done in polynomial time.

It now holds that $s_0 \in \rec(L_1,\dots,L_k)[R]$ if and only if $\varphi := \ram x,y \colon R'(x,y) \vee \exists x \colon R'(x,x)$ is valid where
\[R' := R_p \cap \bigcap_{i=1}^{k} R_i \cap \bigcap_{1 \le i < j \le 2k} R_{i,j} \cap \bigcap_{i=1}^k R_{s_0,i}.\]
Since an ATA for the intersection of two ATAs can be constructed in linear time, we can construct an ATA $\A'$ for $R'$ in time polynomial in the size of the ATAs $\A, \A_1,\dots,\A_k$ and $s_0$. 
By \cref{thm:tree-ramsey} validity of the first disjunct of $\varphi$ can be checked in time exponential in the size of $\A'$.
It is easy to see that validity of the second disjunct of $\varphi$ can also be checked in time exponential in $\A'$.
This yields an exponential time algorithm in total.
\end{proof}

%% file: app-unranked.tex
A \emph{nondeterministic unranked tree automaton} (NUTA) over the unranked alphabet $\Sigma$ is a tuple $\A = (Q,\Sigma,\Delta,q_0)$ where $Q$ and $q_0$ are as in the definition of an NTA and $\Delta \subseteq Q \times \Sigma \times \mathit{REG}(Q)$ is a finite set of transitions. 
Here, $\mathit{REG}(Q)$ denotes the set of regular word languages over $Q$ and we assume that the regular language for each transition is given by an NFA.
A \emph{run} of $\A$ on an unranked tree $t \in \U_\Sigma$ is an unranked tree $\rho \in \U_Q$ with $\dom(\rho) = \dom(t)$ such that $\rho(\varepsilon) = q_0$ and for each inner node $u \in \dom(\rho)$ with children $u1,\dots,ur \in \dom(\rho)$ there is a transition $(\rho(u),t(u),L) \in \Delta$ such that $\rho(u1) \cdots \rho(ur) \in L$. 
A run $\rho$ is \emph{accepting} if for each leaf $u \in \dom(\rho)$ there is a transition $(\rho(u),t(u),L) \in \Delta$ such that $\varepsilon \in L$.
We define \emph{unranked tree-regular relations} in the same way as in the ranked case by using NUTAs instead of NTAs.

\subsection{Proof of Theorem~\ref{thm:unranked-ramsey}}\label{sec:app-unranked-thm}

\begin{definition}
For an unranked tree $t \in \U_\Sigma$ we define the \emph{first-child next-sibling encoding} $\fcns : \N^* \to \{1,2\}^*$ such that 
$\fcns(\varepsilon) = \varepsilon$ and for all $u \in \N^*$ we have $\fcns(u1) = \fcns(u)1$ and $\fcns(u(i+1)) = \fcns(ui)2$ for all $i \ge 1$.
We let $\fcns(\dom(t)) := \bigcup_{u \in \dom(t)} \{\fcns(u),\fcns(u)1,\fcns(u)2\}$.
We define $t' = \fcns(t)$ to be the binary tree with domain $\fcns(\dom(t))$ over the ranked alphabet $\Sigma_\#$ such that 
\begin{itemize}
\item $t'(\fcns(u)) := t(u)$ if $u \in \dom(t)$,
\item $t'(\fcns(u)) := \#$ if $\fcns(u) \in \fcns(\dom(t))$ and $u \notin \dom(t)$.
\end{itemize}
Here, we consider $\Sigma_\# = \Sigma \cup \{\#\}$ as a ranked alphabet with $\rk(a) = 2$ for all $a \in \Sigma$ and $\rk(\#) = 0$.
\end{definition}

\begin{definition}\label{def:convolution}
For unranked trees $t_1, t_2 \in \U_\Sigma$ we define the adapted convolution $t' = t_1 \otimes' t_2 \in \U_{\Sigma_{\#,\bot} \times \Sigma_{\#,\bot}}$ such that $\dom(t') := \dom(t_1) \cup \dom(t_2)$ and
\begin{itemize}
\item $t'(u) := (t_1(u),t_2(u))$ if $u \in \dom(t_1) \cap \dom(t_2)$,
\item $t'(u) := (t_1(u),\#)$ if $u \in \dom(t_1) \setminus \dom(t_2)$ and there exists $v \in \dom(t_2)$ such that $u$ is the first child or right sibling of $v$,
\item $t'(u) := (t_1(u),\bot)$ if $u \in \dom(t_1) \setminus \dom(t_2)$ and the above conditions do not hold,
\item the other cases are symmetric.
\end{itemize}
\end{definition}

\begin{definition}
For the convolution $t = t_1 \otimes' t_2 \in \U_{\Sigma_{\#,\bot} \times \Sigma_{\#,\bot}}$ we define the adapted fist-child next-sibling encoding $t' = \fcns'(t_1 \otimes t_2) \in \T_{\Sigma_{\#,\bot} \times \Sigma_{\#,\bot}}$ such that $\dom(t') := \fcns(\dom(t))$ and for $u' = \fcns(u) \in \fcns(\dom(t))$ we have
\begin{itemize}
\item $t'(u') := t(u)$ if $u \in \dom(t)$,
\item $t'(u') := (\#,\#)$ if $u \notin \dom(t)$ and for parent $\fcns(v)$ of $u'$ we have $t(v) \in \Sigma \times \Sigma$,
\item $t'(u') := (\bot,\#)$ if $u \notin \dom(t)$ and for parent $\fcns(v)$ of $u'$ we have $t(v) \in \{\#,\bot\} \times \Sigma$,
\item $t'(u') := (\#,\bot)$ if $u \notin \dom(t)$ and for parent $\fcns(v)$ of $u'$ we have $t(v) \in \Sigma \times \{\#,\bot\}$.
\end{itemize}
Here, we consider $\Sigma_{\#,\bot} = \Sigma_\# \cup \{\bot\}$ as ranked alphabet with $\rk(\bot) = 0$.
\end{definition}
See \cref{fig:adapted} for an example of the adapted convolution and encoding.

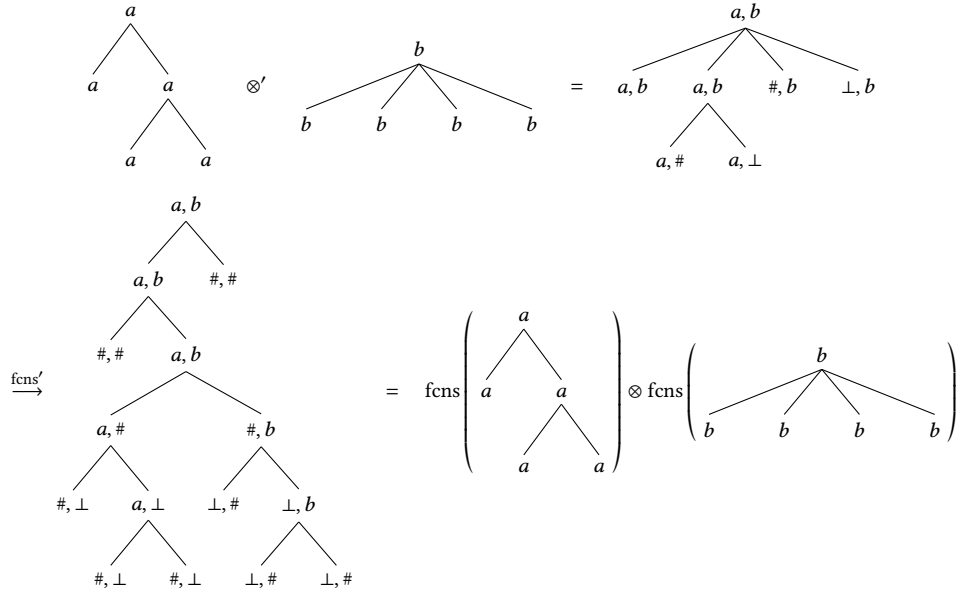
\begin{figure*}
\centering\small
\adjustbox{valign=c}{\begin{tikzpicture}
\node {$a$}
	child { node {$a$} }
	child { node {$a$}
		child { node {$a$}}
		child { node {$a$}}};
\end{tikzpicture}}
\quad $\otimes'$ \quad
\adjustbox{valign=c}{\begin{tikzpicture}
\node {$b$}
	child { node {$b$} }
	child { node {$b$} }
	child { node {$b$} }
	child { node {$b$} };
\end{tikzpicture}}
\quad $=$ \quad
\adjustbox{valign=c}{\begin{tikzpicture}
\node {$a,b$}
	child { node {$a,b$} }
	child { node {$a,b$} 
		child { node {$a,\#$} }
		child { node {$a,\bot$} }}
	child { node {$\#,b$} }
	child { node {$\bot,b$} };
\end{tikzpicture}}
\medskip\\
$\stackrel{\fcns'}{\longrightarrow}$
\adjustbox{valign=c}{\begin{tikzpicture}
\node {$a,b$}
	child { node {$a,b$} 
		child { node {$\#,\#$} }
		child { node {$a,b$} 
			child[sibling distance = 2cm] { node {$a,\#$} 
				child[sibling distance = 1cm] { node {$\#,\bot$} }
				child[sibling distance = 1cm] { node {$a,\bot$} 
					child[sibling distance = 1cm] { node {$\#,\bot$}} 
					child[sibling distance = 1cm] { node {$\#,\bot$}} }}
			child[sibling distance = 2cm] { node {$\#,b$} 
				child[sibling distance = 1cm] { node {$\bot,\#$} }
				child[sibling distance = 1cm] { node {$\bot,b$} 
					child[sibling distance = 1cm] { node {$\bot,\#$}}
					child[sibling distance = 1cm] { node {$\bot,\#$}} }}}}
	child { node {$\#,\#$} };
\end{tikzpicture}}
\quad $=$ \quad
$\fcns \begin{pmatrix}
\adjustbox{valign=c}{\begin{tikzpicture}
\node {$a$}
	child { node {$a$} }
	child { node {$a$}
		child { node {$a$}}
		child { node {$a$}}};
\end{tikzpicture}}
\end{pmatrix}$
$\otimes$
$\fcns \begin{pmatrix}
\adjustbox{valign=c}{\begin{tikzpicture}
\node {$b$}
	child { node {$b$} }
	child { node {$b$} }
	child { node {$b$} }
	child { node {$b$} };
\end{tikzpicture}}
\end{pmatrix}$
\caption{Example for adapted convolution and encoding}
\label{fig:adapted}
\end{figure*}

\begin{lemma}\label{lem:convolution}
Let $R \subseteq \U_\Sigma \times \U_\Sigma$ be an unranked tree-regular relation given by the NUTA $\A$, i.e., $R = \{(t_1,t_2) \in \U_\Sigma \times \U_\Sigma \mid (t_1 \otimes t_2) \in L(\A)\}$.
Then we can construct an NUTA $\A'$ in polynomial time such that $R = \{(t_1,t_2) \in \U_\Sigma \times \U_\Sigma \mid (t_1 \otimes' t_2) \in L(\A')\}$. 
That is, $\A'$ uses the convolution $\otimes'$ instead of $\otimes$.
Conversely, we can also construct $\A$ from $\A'$ in polynomial time such that the above is satisfied.
\end{lemma}
\begin{proof}
To construct $\A'$, we adapt $\A$ such that the conditions of \cref{def:convolution} are satisfied. 
To this end, a state stores for each of the two components if it is the first child or right sibling of a node where the component is labeled by a symbol of $\Sigma$.
Conversely, $\A$ can be constructed by simply replacing $\#$ in the transition of $\A'$ by $\bot$
\end{proof}

The next lemma shows that the connection between the adapted and classical notions of convolution and encoding suggested by \cref{fig:adapted} holds true in general.
\begin{lemma}\label{lem:commutes}
For unranked trees $t_1,t_2 \in \U_\Sigma$ it holds that
\[ \fcns'(t_1 \otimes' t_2) = \fcns(t_1) \otimes \fcns(t_2).\]
\end{lemma}
\begin{proof}
The definitions of $\otimes'$ and $\fcns'$ ensure that the padding symbol $\#$ is used if the node would also be padded by $\fcns$ and otherwise the padding symbol $\bot$ is used. 
The result follows since nodes at the same position in $t_1$ and $t_2$ are mapped to the same position in the encodings $\fcns(t_1)$ and $\fcns(t_2)$.
\end{proof}

The following lemma shows that $\fcns$ preserves all properties of an unranked tree-regular relation and an NTA for the encoded relation over binary trees can be computed in polynomial time.

\begin{lemma}\label{lem:isomorphism}
The encoding $\fcns$ is an isomorphism from a graph $(\U_\Sigma,R)$ of unranked trees to a graph $(\fcns(\U_\Sigma),R')$ of binary trees where $R$ is an unranked tree-regular relation and $R'$ is a tree-regular relation over binary trees. 
Moreover, an NTA that recognizes $R'$ can be constructed in polynomial time given an NUTA recognizing $R$ and vice versa.
\end{lemma}
\begin{proof}
Let $R \subseteq \U_\Sigma \times \U_\Sigma$ be an unranked tree-regular relation given by the NUTA $\A$. By \cref{lem:convolution} one can construct an NUTA $\A'$ that recognizes $R$ but uses $\otimes'$ as convolution. From $\A'$ we can construct an NTA $\B'$ that accepts the language
\begin{multline*}
\{\fcns'(t_1 \otimes' t_2) \mid (t_1 \otimes' t_2) \in L(\A')\} \\= \{\fcns(t_1) \otimes \fcns(t_2) \mid (t_1,t_2) \in R\}
\end{multline*}
where the equality holds by \cref{lem:commutes}.
The automaton $\B'$ can be constructed in polynomial time in the usual way for the first-child next-sibling encoding (cf. \cite{gottlob2005complexity}) but we additionally store in states if the label of the parent of the current node is in $\Sigma \times \Sigma$, $\{\bot,\#\} \times \Sigma$, or $\Sigma \times \{\bot,\#\}$.
Let $R' := \{(\fcns(t_1),\fcns(t_2)) \in \T_{\Sigma_\#} \times \T_{\Sigma_\#} \mid (\fcns(t_1) \otimes \fcns(t_2)) \in L(\B')\}$ be the tree-regular relation over binary trees recognized by $\B'$. Then it hods that $(\fcns(t_1),\fcns(t_2)) \in R'$ iff $(t_1,t_2) \in R$ for any unranked trees $t_1,t_2 \in \U_\Sigma$. Since the first-child next-sibling encoding is injective, it follows that $\fcns$ is an isomorphism from $(\U_\Sigma,R)$ to $(\fcns(\U_\Sigma),R')$.

The construction of an NTA for $R$ given an NTA for $R'$ works analogously using $(\fcns')^{-1}$ and the reverse direction of \cref{lem:convolution}.
\end{proof}

Note that we can generalize the above constructions from the binary to the $n$-ary case such that the same statements hold.

\begin{proof}[Proof of \cref{thm:unranked-ramsey}]
By \cref{lem:isomorphism} we can compute an NTA $\A'$ that recognizes the relation 
\[R' := \{(\fcns(t_1),\dots,\fcns(t_{k+2})) \mid (t_1,\dots,t_{k+2}) \in R\}\] 
over binary trees in polynomial time.
Moreover, it holds that $(\U_\Sigma,R)$ is isomorphic to $(\fcns(\U_\Sigma),R')$.
By \cref{thm:tree-ramsey} we can construct an NTA $\B'$ that recognizes the relation
\[\{\bm{c} \in \fcns(\U_\Sigma)^k \mid \ram x,y \colon R'(x,y,\bm{c})\}\]
in exponential time given $\A'$.
If $R$ and therefore also $R'$ are transitive, \cref{thm:transitive-ramsey} implies that $\B'$ can be computed in polynomial time.
From $\B'$ one can compute an NUTA $\B$ that recognizes the relation
\begin{multline*}
\{(\fcns^{-1}(c_1),\dots,\fcns^{-1}(c_k)) \in (\U_\Sigma)^k \mid \ram x,y \colon R'(x,y,\bm{c})\} \\
= \{\bm{c} \in (\U_\Sigma)^k \mid \ram x,y \colon R(x,y,\bm{c})\}
\end{multline*}
in polynomial time by applying the reverse direction of \cref{lem:isomorphism}.
\end{proof}

\subsection{Definition of subtree and flat prefix rewriting systems}\label{sec:app-sfprs}
For a tree $t  \in \U_\Sigma$ and node $x \in \dom(t)$ we write $t_{\downarrow x}$ for the subtree of $t$ rooted in $x$. 
We denote by $t[x|s]$ the tree that is obtained from $t$ if we replace $t_{\downarrow x}$ by the tree $s \in \U_\Sigma$.
If $\mathrm{ht}(t) = 1$, we denote the sequence of leaves of $t$ read from left to right by $\mathrm{flatfront}(t)$.
Here, $\mathrm{ht}(t) := \max\{|x| \mid x \in \dom(t)\}$ is defined as the height of $t$.

\begin{definition}
A \emph{subtree and flat prefix rewriting system (SFPRS)} over unranked trees in $\U_\Sigma$ is of the form $\mathcal{R} = (\Sigma,\Gamma,R,t_{\mathrm{in}})$, with a finite unranked alphabet $\Sigma$, a finite transition alphabet $\Gamma$, an initial tree $t_{\mathrm{in}}$, and a finite set $R$ of rules of two types:
\begin{enumerate}
\item subtree substitution with rules of the form 
$r_j \colon s_j \rewrite{\sigma} s'_j$ for $j \in J$, $s_j, s'_j \in \U_\Sigma$, $\sigma \in \Gamma$, and
\item flat prefix substitution at the flat front of the tree with rules of the form
$r_i \colon u_i \rewrite{\sigma} u'_i$ for $i \in I$, $u_i, u'_i \in \Sigma^+$, $\sigma \in \Gamma$,
\end{enumerate}
with $I \cup J = \{1,\dots,|R|\}$ and $I \cap J = \emptyset$.

A tree $t'$ is derived from $t$ (denoted $t \to_\mathcal{R}^\sigma t'$) by applying a subtree rewrite rule $r_j$, if there is a node $x \in \dom(t)$ with $t_{\downarrow x} = s_j$ such that $t[x | s'_j] = t'$.

A tree $t'$ is derived from $t$ by applying a prefix rewrite rule $r_i$, if there is a node $x \in \dom(t)$ with $\mathrm{ht}(t_{\downarrow x}) = 1$ and $\mathrm{flatfront}(t_{\downarrow x}) = u_i v$ and a tree $s \in \U_\Sigma$ with $\mathrm{ht}(s) = 1$, $s(\varepsilon) = t(x)$, and $\mathrm{flatfront}(s) = u'_i v$ such that $t[x | s] = t'$ for some $v \in \Sigma^*$.
\end{definition}

The definition of an SFPRS can be extended to a \emph{regular SFPRS} by allowing subtree rewrite rules of the form $S_j \rewrite{\sigma} S'_j$ with unranked tree-regular languages $S_j, S'_j \subseteq \U_\Sigma$ and prefix rewrite rules of the form $L_i \rewrite{\sigma} L'_i$ with regular languages $L_i, L'_i \subseteq \Sigma^*$.
Clearly, SFPRSs are special regular SFPRSs where the rules only have singleton sets.